\pdfoutput=1
\documentclass[acmsmall,screen,nonacm]{acmart}

\usepackage{amsmath}
\usepackage{algorithm}
\usepackage{algorithmic}
\usepackage{graphicx}
\usepackage{booktabs}
\usepackage{tabularx}
\usepackage{multirow}
\usepackage{subcaption}
\usepackage{xspace}
\usepackage{listings}
\usepackage{tikz}
\usetikzlibrary{arrows.meta,positioning,shapes.geometric,fit}

\setcopyright{none}
\renewcommand\footnotetextcopyrightpermission[1]{}
\newcommand{\sysname}{AccelSync\xspace}
\newcommand{\apf}{accelerator pipeline\xspace}
\newcommand{\hb}{\ensuremath{\rightarrow_{\mathit{hb}}}}
\newcommand{\po}{\ensuremath{\rightarrow_{\mathit{po}}}}
\newcommand{\so}{\ensuremath{\rightarrow_{\mathit{so}}}}
\newcommand{\bo}{\ensuremath{\rightarrow_{\mathit{bo}}}}

\begin{document}

\title{AccelSync: Verifying Synchronization Coverage in Accelerator Pipeline Programs}
\author{Hangcheng An}
\affiliation{%
  \institution{Beihang University}
  \city{Beijing}
  \country{China}
}

\author{Rui Wang}
\email{wangrui@buaa.edu.cn}
\affiliation{%
  \institution{Beihang University}
  \city{Beijing}
  \country{China}
}

\author{Depei Qian}
\affiliation{%
  \institution{Beihang University}
  \city{Beijing}
  \country{China}
}
\authorsaddresses{Corresponding author: Rui Wang, wangrui@buaa.edu.cn}

\begin{abstract}
AI accelerator operators are compiled into multi-stage pipeline programs where DMA, vector, matrix, and scalar units execute concurrently on shared on-chip buffers.
A missing or misplaced synchronization primitive introduces hardware-visible data races that escape both simulation and golden testing, because neither models the accelerator's cross-unit visibility semantics.
We formalize accelerator pipeline programs as a restricted concurrent language, define a parameterized hardware event semantics with three ordering relations---program order, synchronization order, and barrier order---and reduce the correctness question to \emph{barrier sufficiency}---equivalently, \emph{synchronization coverage}---whether every cross-unit write-read pair on the same buffer is ordered by happens-before.
Here ``barrier'' denotes an abstract ordering primitive in the model, covering vendor pipe barriers, hard-event synchronization, and equivalent frontend-normalized synchronization points.
We prove that synchronization coverage is decidable in $O(|E|^2)$ time and that our checker is both sound and complete under the modeled semantics.
We implement \sysname, a static verification tool instantiated for Ascend 910B2 and Cambricon MLU370 by changing only the hardware model.
On 6{,}292 production kernels from the CANN operator library, \sysname identifies 3 previously unknown synchronization hazards---one matching a hazard class for which we observed nondeterministic outputs on Ascend 910B2 under a specific toolkit/driver configuration (CANN 8.0.RC3), though this observation was not reproducible after a subsequent driver upgrade---and on 120 LLM-generated kernels it flags a 19.2\% defect rate (95\% CI: [13.0\%, 27.4\%]).
A mutation study on 688 non-equivalent mutants yields 100\% detection, and a head-to-head comparison shows \sysname detects hazards that Huawei's runtime sanitizer msSanitizer misses, at 400$\times$ lower cost per kernel.
\end{abstract}

\ccsdesc[500]{Software and its engineering~Formal software verification}
\ccsdesc[300]{Computer systems organization~Data flow architectures}
\keywords{accelerator verification, synchronization, happens-before, Ascend C, static analysis}

\maketitle

\section{Introduction}\label{sec:introduction}

Modern AI accelerators---including Google TPUs~\cite{jouppi-tpu}, Graphcore IPUs~\cite{graphcore-ipu}, Huawei Ascend NPUs~\cite{ascendc}, and Cambricon MLUs~\cite{bangc}---organize operator execution as multi-stage pipelines where DMA engines, vector processing units, matrix engines, and scalar cores operate concurrently on shared on-chip buffers. Operator compilers such as Triton~\cite{triton}, TVM~\cite{tvm}, Halide~\cite{halide}, AKG~\cite{akg}, and vendor-specific toolchains (Ascend C, BANG C) translate high-level tensor programs into this pipeline code. The translation may preserve functional correctness while still violating hardware synchronization requirements: a producer unit writes a buffer, a consumer unit reads it, and the required synchronization coverage is missing. These hazards are hard to catch because x86 simulation and golden testing do not model the accelerator's cross-unit visibility rules~\cite{mssanitizer,ascendrs}.

Existing verification approaches leave this gap largely unaddressed. GPU verification tools such as GPUVerify~\cite{gpuverify}, WEFT~\cite{weft}, and GKLEE~\cite{gklee} target flat warp-level synchronization with \texttt{\_\_syncthreads} and named barriers---a fundamentally different concurrency model from hierarchical operator pipelines where distinct hardware units communicate through queues, pipe drains, and hard-event synchronization. Hardware memory model formalizations~\cite{ptxmem,alglave,wickerson-memalloy} reason about instruction-level ordering on CPUs and GPUs but do not address operator-level pipeline synchronization. LLM-based verification frameworks such as ProofWright~\cite{proofwright} and FM-Agent~\cite{fmagent} reason about sequential program semantics without awareness of implicit hardware visibility requirements. Dynamic race detectors~\cite{eraser,flanagan-fasttrack,tsan} and vendor sanitizers~\cite{mssanitizer,cuda-racecheck} require runtime instrumentation and hardware access, limiting their applicability to large-scale static audits. As a result, synchronization hazards in accelerator pipelines often remain invisible until they manifest on real devices.

This paper studies a narrower but practically important problem: \emph{is the synchronization present in a structured accelerator kernel sufficient for the hardware on which it will run?} We answer this question with \sysname, a lightweight static checker for structured accelerator pipelines. \sysname extracts memory events from a kernel, instantiates a parameterized hardware model describing execution units and synchronization primitives, builds a happens-before graph following the event-based reasoning style of Lamport~\cite{lamport} and Alglave et al.~\cite{alglave}, and reports write-read pairs not covered by the available synchronization.

The key insight is that accelerator operator pipelines occupy a \emph{structured fragment} of concurrent programs---fixed stages, queue-mediated communication, and a small set of hardware synchronization primitives---that makes synchronization-coverage checking decidable and efficient. We formalize this class of \emph{\apf programs} and prove that barrier sufficiency---the formal criterion underlying synchronization coverage---is decidable in $O(|E|^2)$ with soundness and completeness under the modeled hardware event semantics. In our model, ``barrier'' is an abstract ordering primitive: on Ascend C it can represent pipe barriers, \texttt{SetFlag}/\texttt{WaitFlag} hard-event synchronization, or equivalent frontend-normalized synchronization points. \sysname exploits this structure to establish a \emph{cross-layer contract}: the compiler-generated pipeline topology constrains the concurrency, and the hardware visibility model constrains the synchronization requirements. The checker verifies that the synchronization inserted by the compiler satisfies the requirements imposed by the hardware. This cross-layer perspective distinguishes \sysname from GPU verification tools (which assume flat SIMT concurrency~\cite{gpuverify,weft,autosync}), sequential verification frameworks (which ignore hardware visibility~\cite{proofwright,fmagent}), and abstract-interpretation-based analyzers (which target different program properties~\cite{cousot-cousot,mine-octagon}).

In summary, this paper makes the following contributions:
\begin{itemize}
\item We identify synchronization coverage in structured accelerator pipelines as a distinct verification problem that is not captured by simulation testing, sequential verification, or prior GPU-oriented barrier checkers (\S\ref{sec:background}).
\item We formalize \apf programs as a restricted concurrent language, define hardware event semantics with three ordering relations (po, so, bo), and prove that barrier sufficiency (synchronization coverage) is decidable in $O(|E|^2)$ with soundness and completeness under the modeled semantics (\S\ref{sec:formalization}).
\item We implement \sysname with three frontends (Ascend~C source-level, Ascend~C IR-level, and BANG~C source-level) and parameterized hardware models for Ascend 910B2 and Cambricon MLU370 (\S\ref{sec:implementation}).
\item We evaluate \sysname on 6{,}412 Ascend~C kernels across two code sources---6{,}292 production CANN kernels and 120 LLM-generated kernels---finding synchronization hazards in both, with risk rates ranging from 0.048\% (production) to 19.2\% (LLM-generated). On 688 injected mutations, \sysname achieves 100\% detection. A head-to-head comparison with Huawei's msSanitizer shows \sysname detects hazards that runtime instrumentation misses, at 400$\times$ lower per-kernel cost. Hardware testing of a representative VPU$\to$scalar hazard class on Ascend 910B2 (CANN 8.0.RC3) observed nondeterministic outputs consistent with the predicted failure mode, though this was not reproducible after a driver upgrade. Cross-hardware instantiation on Cambricon MLU370 audits 162 BANG~C kernels in 393\,ms, confirming portability via parameter substitution (\S\ref{sec:evaluation}).
\end{itemize}

The full CANN corpus verifies in under 48 seconds, and the largest single kernel (274 events) verifies in 21\,ms, confirming that static synchronization-coverage auditing is practical at production scale.

The rest of this paper is organized as follows. \S\ref{sec:background} introduces accelerator pipeline architecture and motivates the problem with a concrete hazard example. \S\ref{sec:formalization} formalizes \apf programs and proves decidability, soundness, and completeness. \S\ref{sec:implementation} describes the implementation. \S\ref{sec:evaluation} presents the evaluation. \S\ref{sec:related} discusses related work, and \S\ref{sec:conclusion} concludes.

\section{Background and Motivation}\label{sec:background}

We describe the pipeline architecture of AI accelerators, demonstrate how synchronization hazards arise and evade testing, and position \sysname relative to prior verification approaches.

\subsection{Accelerator Pipeline Architecture}\label{sec:bg-arch}

Modern AI accelerators organize operator execution as a multi-stage pipeline to overlap data movement with computation. On Ascend 910B2, the pipeline consists of three stages: MTE\_in (memory transfer engine, loading data from global memory to on-chip unified buffer), Compute (vector/matrix/scalar processing on the unified buffer), and MTE\_out (writing results back to global memory). Stages communicate through hardware FIFO queues called TQues---the producer stage issues an EnQue to signal data availability, and the consumer stage issues a DeQue to wait for it.

\begin{figure}[t]
\centering
\includegraphics[width=\linewidth]{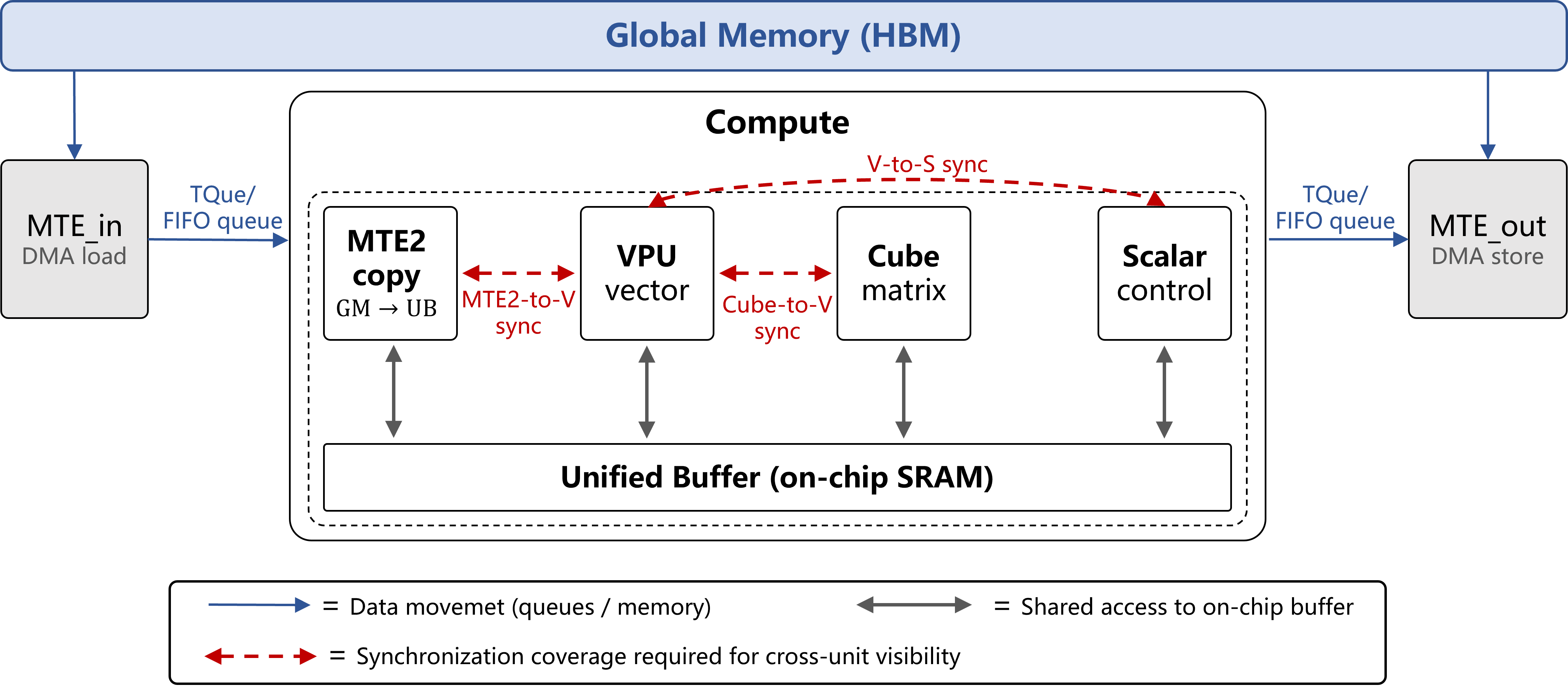}
\caption{Ascend 910B2 pipeline architecture. Three stages communicate via FIFO queues (TQue). Within the Compute stage, the MTE2 copy pipe, VPU, Cube, and Scalar units share the unified buffer but execute independently; cross-unit visibility requires explicit synchronization coverage such as hard events, pipe drains, or equivalent compiler-normalized primitives (red dashed lines). The MTE2$\to$V edge denotes copy-pipe-to-vector coverage for data copied into UB before vector consumption, not stage-level queue synchronization.}
\Description{Pipeline architecture with global memory, MTE input and output stages, compute units, unified buffer, and dashed synchronization coverage edges between MTE2, VPU, Cube, and Scalar units.}
\label{fig:architecture}
\end{figure}

Within the Compute stage, three execution units operate concurrently: the Vector Processing Unit (VPU) for elementwise and reduction operations, the Cube unit for matrix multiplication, and the Scalar unit for control flow and index computation. These units share the unified buffer but execute independently---a VPU write does not automatically become visible to a subsequent scalar read. Ascend C exposes both pipe-drain barriers such as \texttt{PipeBarrier<PIPE\_V>()} and directional hard-event synchronization such as \texttt{SetFlag}/\texttt{WaitFlag<HardEvent::V\_S>}. In the rest of the paper we use \emph{synchronization coverage} to mean that the producer-consumer pair is ordered by the appropriate vendor primitive or an equivalent frontend-normalized synchronization point.

Cambricon MLU370 follows a similar three-stage model (IO\_in, Compute, IO\_out) with DMA for data movement, VPU for vector operations, and IPU for matrix operations. Synchronization uses \texttt{\_\_sync()} as a full barrier and \texttt{\_\_sync\_io()} for DMA-to-compute visibility, with queue-based inter-stage communication analogous to TQues.

\subsection{The Synchronization Problem}\label{sec:bg-problem}

We illustrate the problem with a concrete example from a \texttt{softmax} operator kernel targeting Ascend 910B2. Listing~\ref{lst:softmax-bug} shows the relevant fragment.

\begin{figure}[t]
\begin{lstlisting}[
  language=C++,
  basicstyle=\ttfamily\scriptsize,
  numbers=left, numberstyle=\tiny,
  xleftmargin=2em,
  frame=single,
  caption={Simplified Ascend C fragment from the \texttt{softmax} operator. The VPU write at line~4 and the scalar read at line~5 execute on different hardware units. Without V-to-S synchronization coverage between them, the scalar unit may read a stale value.},
  label={lst:softmax-bug}
]
// Stage: Compute (on unified buffer)
LocalTensor<float> tmp = inQueueX.DeQue<float>();
ReduceMax(maxVal, tmp, 0, reduceLen); // VPU write
// HAZARD: missing V-to-S synchronization coverage here
float s = maxVal.GetValue(0);          // Scalar read
// ... normalization using s ...
outQueueY.EnQue<float>(result);
\end{lstlisting}
\Description{Simplified Ascend C softmax fragment where a VPU ReduceMax write is followed by a scalar GetValue read without V-to-S synchronization coverage.}
\end{figure}

The kernel computes row-wise maximum via \texttt{ReduceMax} (VPU), then extracts the scalar result via \texttt{GetValue} (scalar unit) for subsequent normalization. The compiled code omits synchronization coverage between these operations, such as a directional V-to-S hard event or an equivalent V-pipe drain recognized by the frontend. On x86 simulation, all operations execute sequentially on a single thread---the VPU write completes before the scalar read, and the test passes. On real hardware, VPU and scalar execute concurrently; the scalar unit may read a stale value from the buffer before the VPU write lands.

This hazard class---VPU write followed by scalar read without V-to-S synchronization coverage---is not an isolated incident. A second class---cube write followed by VPU read without cube-to-V synchronization coverage---is more severe because the pipeline provides no heuristic coverage for cube-to-VPU transitions. Both classes are invisible to x86 simulation and golden testing, and they persist in released code until triggered on real hardware. Our evaluation (\S\ref{sec:evaluation}) quantifies the prevalence of these hazard classes across 6{,}412 Ascend~C kernels and 162 BANG~C kernels.

\subsection{Taxonomy of Prior Approaches}\label{sec:bg-taxonomy}

\begin{table*}[t]
\caption{Comparison of verification approaches for accelerator concurrency. ``Sound'' refers to each method's own target property; sequential verifiers remain orthogonal because they check functional correctness rather than hardware visibility.}\label{tab:taxonomy}
\centering
\small
\setlength{\tabcolsep}{4pt}
\begin{tabular}{@{}lcccc@{}}
\toprule
\textbf{Approach} & \textbf{Concurrency} & \textbf{Sync Model} & \textbf{HW-Aware} & \textbf{Sound} \\
\midrule
GPUVerify~\cite{gpuverify} & Warps & \_\_syncthreads & GPU only & Yes \\
WEFT~\cite{weft} & Warps & Named barriers & GPU only & Yes \\
AUTOSYNC~\cite{autosync} & Warps & Synthesis & GPU only & Yes \\
PTX Mem~\cite{ptxmem} & Instructions & Mem. model & GPU only & Yes \\
ProofWright~\cite{proofwright} & Sequential & Hoare logic & No & Partial \\
FM-Agent~\cite{fmagent} & Sequential & Hoare logic & No & Partial \\
Golden Test & None & N/A & No & No \\
\midrule
\textbf{\sysname} & \textbf{Pipeline} & \textbf{Queue+sync coverage} & \textbf{Yes} & \textbf{Yes$^\dagger$} \\
\bottomrule
\multicolumn{5}{@{}p{\dimexpr\linewidth-2\tabcolsep\relax}@{}}{\footnotesize $^\dagger$Sound and complete for \apf kernels under the vendor-documented hardware event semantics we model.}
\end{tabular}
\end{table*}

Table~\ref{tab:taxonomy} summarizes the landscape. GPU verification tools model flat warp-level concurrency with barrier-based synchronization. They assume all threads execute the same program (SIMT) and synchronize through shared barriers---a fundamentally different model from hierarchical pipelines where distinct hardware units communicate through queues. Hardware memory model work formalizes instruction-level ordering on GPUs but does not address operator-level pipeline synchronization. LLM-based verification reasons about sequential correctness without modeling hardware concurrency. Simulation testing eliminates concurrency entirely by executing stages sequentially.

\sysname fills the gap by formalizing the specific concurrency model of accelerator pipelines---multi-stage, queue-synchronized, with intra-stage unit concurrency---and providing decidable verification with soundness and completeness guarantees.

\section{The Accelerator Pipeline Model and Verification Theory}\label{sec:formalization}

We now formalize the synchronization structure of accelerator operator pipelines as a restricted concurrent language, define hardware event semantics over this language, and prove that synchronization coverage---formally termed \emph{barrier sufficiency}---is decidable with polynomial complexity. The term \emph{barrier} is used as an abstract ordering primitive in this section; concrete backends may realize it with pipe barriers, hard-event \texttt{SetFlag}/\texttt{WaitFlag} pairs, queue synchronization, or equivalent compiler-normalized synchronization points.

\subsection{Model Scope and Assumptions}\label{sec:modelscope}

The theorems in this section apply to kernels that can be faithfully lowered into \apf under a vendor-documented event semantics. Concretely, our model assumes: (1) queue operations follow FIFO discipline; (2) the relevant visibility effects are exactly those induced by per-unit program order, queue synchronization, and documented barrier or hard-event primitives; (3) timing nondeterminism arises from interleavings among execution units rather than speculative execution or undocumented coherence mechanisms; and (4) buffer identities used in the event graph are already resolved by the frontend, i.e., aliasing is not introduced after extraction.

These assumptions are deliberate rather than incidental. They match the structured kernel fragment targeted by \sysname, and they explain why our soundness and completeness results are necessarily scoped to the modeled event semantics rather than to the physical hardware in full generality. Kernels with runtime-computed synchronization identifiers, data-dependent pipeline topologies, or undocumented ordering sources fall outside this scope and are treated conservatively by the implementation.

\subsection{Accelerator Pipeline Programs}\label{sec:apf}

\begin{figure}[t]
\centering
\includegraphics[width=\linewidth]{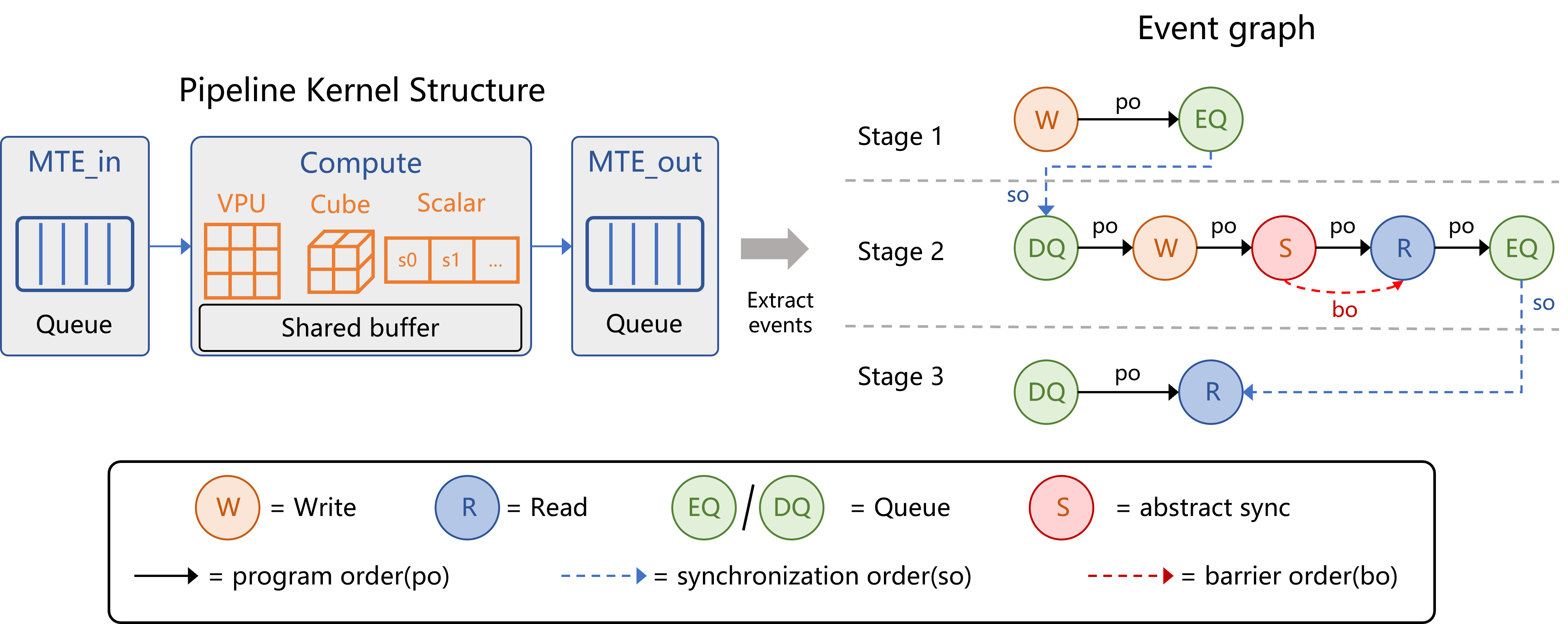}
\caption{From accelerator pipeline kernel structure to event graph. The frontend extracts typed events (Write, Read, EnQue, DeQue, abstract synchronization) from each pipeline stage and connects them with program-order edges (solid arrows) within each stage and synchronization-order edges (dashed arrows) across stages via queue matching. This event graph is the input to Algorithm~\ref{alg:check}.}
\Description{Pipeline stages are converted into an event graph with write, read, enqueue, dequeue, and abstract synchronization events connected by program order, synchronization order, and barrier order edges.}
\label{fig:apf-to-events}
\end{figure}

Accelerator operators execute as multi-stage pipelines where each stage runs on a dedicated execution unit and communicates through on-chip buffer queues. We capture this structure as a class of \emph{\apf programs}, a restricted concurrent language defined by the following grammar:

\begin{align}
\mathit{Program} &::= \mathit{Stage}_1 \parallel \mathit{Stage}_2 \parallel \cdots \parallel \mathit{Stage}_n \label{eq:program} \\
\mathit{Stage} &::= \mathit{Op}^* \label{eq:stage} \\
\mathit{Op} &::= \mathrm{Write}(b) \mid \mathrm{Read}(b) \mid \mathrm{EnQue}(q) \mid \mathrm{DeQue}(q) \mid \mathrm{Barrier}(p) \label{eq:op} \\
\mathit{Topology} &::= (\mathit{Stage}_i, \mathit{Stage}_j, q)^* \label{eq:topo}
\end{align}

where $b$ ranges over buffer names, $q$ over queue identifiers, and $p$ over pipe identifiers. The topology specifies directed queue connections between stages.

Accelerator pipeline programs impose five structural restrictions that distinguish them from general concurrent programs:

\begin{description}
\item[R1 (Finite Stages)] The number of stages $n$ is fixed at compile time. For Ascend 910B2, $n = 3$ (MTE\_in, Compute, MTE\_out); for Cambricon MLU370, $n = 3$ (IO\_in, Compute, IO\_out).

\item[R2 (Sequential Stages)] Each stage executes its operations sequentially. Within the Compute stage, multiple execution units (VPU, cube, scalar) may operate concurrently, but each unit's instruction stream is sequential.

\item[R3 (Queue-and-Barrier Synchronization)] Inter-stage communication occurs exclusively through FIFO queues (EnQue/DeQue pairs) and modeled synchronization primitives. There is no shared-memory communication outside these mechanisms.

\item[R4 (Bounded Queues)] Each queue has a compile-time-known depth. The queue discipline is strict FIFO---the $k$-th EnQue on queue $q$ pairs with the $k$-th DeQue on $q$.

\item[R5 (No Dynamic Concurrency)] There is no dynamic thread creation, join, or fork. The pipeline topology is static throughout execution.
\end{description}

These restrictions hold for all operators we examined across Ascend 910B2 and Cambricon MLU370. They reflect a fundamental property of accelerator operator pipelines---the concurrency is \emph{structured} rather than general-purpose, arising from the fixed hardware datapath rather than programmer-specified parallelism, which sharply distinguishes this setting from GPU SIMT verification and symbolic-execution-based GPU hazard detection~\cite{gpuverify,weft,gklee}.

\subsection{Hardware Event Semantics}\label{sec:events}

Given an \apf program $P$ and a hardware model $H$, we define the set of memory events and three ordering relations that together determine the happens-before relation, following the standard style of event-based memory-model reasoning~\cite{lamport,alglave,ptxmem}.

\subsubsection{Events}

The event set $E$ consists of all memory operations extracted from $P$:
\begin{equation}
E = \{ e = (\mathit{kind}, b, s, t, u) \}
\end{equation}
where $\mathit{kind} \in \{\mathrm{W}, \mathrm{R}, \mathrm{EQ}, \mathrm{DQ}, \mathrm{B}\}$ denotes write, read, enqueue, dequeue, or barrier; $b$ is the buffer or queue name; $s$ is the stage; $t$ is the logical timestamp within the stage; and $u$ is the execution unit. For regular reads and writes, $u$ is one of the hardware units (e.g., VPU, cube, scalar). For queue events (EQ/DQ), we treat $u$ as undefined: they are stage-level synchronization events rather than unit-local actions.

\subsubsection{Program Order (po)}

Program order captures the sequential execution guarantee within each execution unit. For events $e_1, e_2$ in the same stage $s$:
\begin{equation}\label{eq:po}
e_1 \po e_2 \iff u(e_1) = u(e_2) \wedge t(e_1) < t(e_2)
\end{equation}

This definition reflects a critical hardware property---within the Compute stage, VPU, cube, and scalar units execute \emph{concurrently}. A VPU write does not automatically become visible to a subsequent scalar read; that requires an explicit barrier. Program order only guarantees visibility within the same unit for unit-local events. Queue events are the only exception: because they synchronize stage progress rather than a single unit's pipeline, they induce stage-wide po edges as an explicit extension of Eq.~(\ref{eq:po}).

EnQue and DeQue operations serve as stage-global synchronization points. A DeQue event happens-before all subsequent events in the same stage regardless of unit, and all prior events happen-before an EnQue:
\begin{align}
e_{\mathrm{DQ}} \po e &\iff s(e_{\mathrm{DQ}}) = s(e) \wedge t(e_{\mathrm{DQ}}) < t(e) \label{eq:deque-po} \\
e \po e_{\mathrm{EQ}} &\iff s(e) = s(e_{\mathrm{EQ}}) \wedge t(e) < t(e_{\mathrm{EQ}}) \label{eq:enque-po}
\end{align}

\subsubsection{Synchronization Order (so)}

Synchronization order captures the cross-stage visibility established by queue operations. For queue $q$ connecting stage $s_1$ to stage $s_2$:
\begin{equation}\label{eq:so}
\mathrm{EQ}_k(q, s_1) \so \mathrm{DQ}_k(q, s_2)
\end{equation}
where the subscript $k$ denotes the $k$-th operation on queue $q$, matched by FIFO order (restriction R4).

\subsubsection{Barrier Order (bo)}

Barrier order captures the intra-stage visibility established by modeled synchronization primitives. The event $\mathrm{B}(p,s,t_b)$ is an abstract ordering event, not necessarily the literal Ascend C \texttt{PipeBarrier} API: the frontend may lower \texttt{PipeBarrier}, \texttt{SetFlag}/\texttt{WaitFlag}, or an equivalent synchronization idiom to such an event. The hardware model $H$ defines a set of synchronization-coverage rules, each specifying which unit-to-unit visibility a modeled primitive provides. For a coverage rule $(p, s, u_w, u_r) \in H$:
\begin{equation}\label{eq:bo}
\mathrm{W}(b, s, t_w, u_w) \bo \mathrm{R}(b, s, t_r, u_r) \iff \exists\, \mathrm{B}(p, s, t_b) : t_w < t_b < t_r
\end{equation}

On Ascend 910B2, the model uses directional synchronization coverage rules such as VPU$\to$scalar, cube$\to$VPU, and MTE$\to$VPU. In Ascend C terminology, cross-pipe coverage is commonly expressed with \texttt{HardEvent} synchronization such as \texttt{HardEvent::V\_S}, while pipe drains such as \texttt{PipeBarrier<PIPE\_V>()} can serve as implementation-level evidence when the frontend or compiler log establishes equivalent coverage. On Cambricon MLU370, $\mathrm{\_\_sync()}$ acts as a full barrier establishing all-to-all visibility within the Compute stage.

\begin{table}[t]
\caption{Hardware model parameters for Ascend 910B2 and Cambricon MLU370. Each rule specifies the required producer-consumer coverage and representative primitives that the frontend may normalize into the abstract barrier-order relation.}\label{tab:hw-models}
\centering
\small
\setlength{\tabcolsep}{3pt}
\begin{tabularx}{\linewidth}{@{}l l >{\raggedright\arraybackslash}X l@{}}
\toprule
\textbf{Platform} & \textbf{Producer $\to$ Consumer} & \textbf{Representative primitive} & \textbf{Source} \\
\midrule
\multirow{3}{*}{Ascend 910B2}
  & VPU $\to$ scalar & \texttt{HardEvent::V\_S} / V-pipe drain & AscendC API~\cite{ascendc} \\
  & Cube $\to$ VPU & cube-to-V sync / equivalent drain & AscendC API~\cite{ascendc} \\
  & MTE $\to$ VPU & \texttt{HardEvent::MTE2\_V} / MTE2 drain & AscendC API~\cite{ascendc} \\
\midrule
\multirow{4}{*}{MLU370}
  & VPU $\to$ IPU & \texttt{\_\_sync()} & BANG C Guide \\
  & IPU $\to$ VPU & \texttt{\_\_sync()} & BANG C Guide \\
  & DMA $\to$ VPU/IPU & \texttt{\_\_sync\_io()} & BANG C Guide \\
  & VPU/IPU $\to$ DMA & \texttt{\_\_sync\_compute()} & BANG C Guide \\
\bottomrule
\end{tabularx}
\end{table}

Table~\ref{tab:hw-models} summarizes the hardware model parameters for both platforms. The Ascend 910B2 model uses fine-grained directional coverage rules, while the MLU370 model uses coarse-grained full barriers. This difference is captured entirely by the barrier rule set---the core checking algorithm remains identical. Adding a new hardware backend requires only specifying the stages, queues, and synchronization-coverage rules; no changes to the checker are needed.

\subsubsection{Happens-Before}

The happens-before relation is the transitive closure of the union of all three orderings:
\begin{equation}\label{eq:hb}
\hb\; = (\po \cup \so \cup \bo)^+
\end{equation}

\subsection{Barrier Sufficiency}\label{sec:sufficiency}

\begin{figure}[t]
\centering
\includegraphics[width=0.5\linewidth]{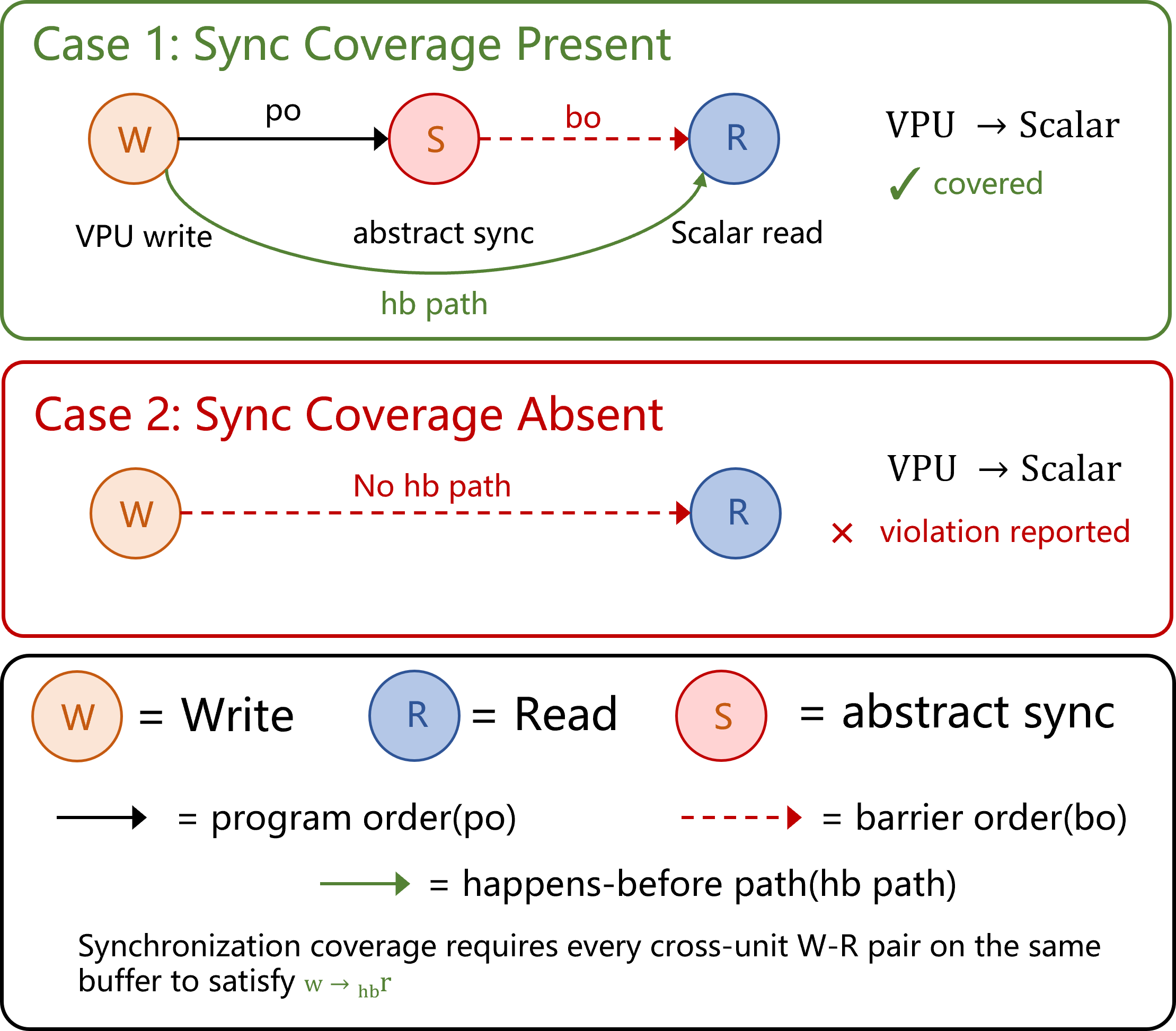}
\caption{Synchronization coverage (barrier sufficiency) illustrated at the abstract event-graph level. \emph{Top:} A VPU write reaches a Scalar read via program order and a modeled synchronization edge, yielding a happens-before path (safe). \emph{Bottom:} Without synchronization coverage, no ordering exists between the write and read---AccelSync reports this as a synchronization violation.}
\Description{Two event-graph cases: one with synchronization coverage forming a happens-before path from a write to a read, and one without such a path reporting a violation.}
\label{fig:hb-sufficiency}
\end{figure}

We define the correctness condition for an \apf program as barrier sufficiency (synchronization coverage)---every cross-unit Write-Read pair on the same buffer must be ordered by happens-before through some modeled synchronization.

\begin{definition}[Barrier Sufficiency]\label{def:sufficiency}
An \apf program $P$ is \emph{barrier-sufficient} under hardware model $H$ if and only if for all events $w, r \in E$ where $\mathit{kind}(w) = \mathrm{W}$, $\mathit{kind}(r) = \mathrm{R}$, $b(w) = b(r)$, and either $s(w) \neq s(r)$ or ($s(w) = s(r) \wedge u(w) \neq u(r) \wedge t(w) < t(r)$):
\begin{equation}
w \hb r
\end{equation}
\end{definition}

A program that is not barrier-sufficient contains at least one Write-Read pair where the write may not be visible to the read---a potential data race under the hardware's concurrent execution model.

\subsection{Decidability}\label{sec:decidability}

Algorithm~\ref{alg:check} summarizes the complete verification procedure.

\begin{algorithm}[t]
\caption{Barrier Sufficiency Checking}\label{alg:check}
\begin{algorithmic}[1]
\REQUIRE \apf program $P$, hardware model $H$
\ENSURE Set of uncovered Write-Read pairs (empty $\Rightarrow$ barrier-sufficient)
\STATE $E \leftarrow \textsc{ExtractEvents}(P)$ \COMMENT{Parse kernel into event set}
\STATE $G \leftarrow$ empty directed graph over $E$
\FORALL{stage $s$, unit $u$}
  \STATE Add \textbf{po} edges between consecutive events in $(s, u)$
  \STATE Add \textbf{po} edges from DeQue to all later events in $s$
  \STATE Add \textbf{po} edges from all earlier events in $s$ to EnQue
\ENDFOR
\FORALL{queue $q$ connecting $s_1 \to s_2$}
  \STATE Match $\mathrm{EQ}_k(q, s_1)$ with $\mathrm{DQ}_k(q, s_2)$ by FIFO order
  \STATE Add \textbf{so} edge for each matched pair
\ENDFOR
\FORALL{coverage rule $(p, s, u_w, u_r) \in H$}
  \FORALL{$\mathrm{B}(p, s, t_b)$ in $E$}
    \STATE Add \textbf{bo} edges: $\mathrm{W}(u_w, t < t_b) \to \mathrm{R}(u_r, t > t_b)$
  \ENDFOR
\ENDFOR
\STATE $\mathit{Reach} \leftarrow \textsc{TransitiveClosure}(G)$ \COMMENT{BFS from each node, $O(|E|^2)$}
\STATE $V \leftarrow \emptyset$
\FORALL{$(w, r)$ where $\mathit{kind}(w){=}\mathrm{W}$, $\mathit{kind}(r){=}\mathrm{R}$, $b(w){=}b(r)$, cross-unit}
  \IF{$w \not\in \mathit{Reach}(r)$}
    \STATE $V \leftarrow V \cup \{(w, r)\}$
  \ENDIF
\ENDFOR
\RETURN $V$
\end{algorithmic}
\end{algorithm}

\begin{theorem}[Decidability of Barrier Sufficiency]\label{thm:decidability}
Given an \apf program $P$ satisfying restrictions R1--R5 and a hardware model $H$, the barrier sufficiency problem is decidable in $O(|E|^2)$ time, where $|E|$ is the number of memory events in $P$.
\end{theorem}

\begin{proof}[Proof sketch]
The argument proceeds in three steps:

\emph{Step 1: Finiteness of the event set.} Restrictions R1 (finite stages), R2 (sequential execution), and R5 (no dynamic concurrency) guarantee that $|E|$ is finite and statically determinable from the program text.

\emph{Step 2: Computability of happens-before.} The relations po, so, and bo are each computable from $E$ and $H$ in $O(|E|)$ time (po by sorting within each unit; so by FIFO matching per queue; bo by scanning for barrier events). The transitive closure $(\po \cup \so \cup \bo)^+$ is computable in $O(|E|^2)$ via BFS from each node in the event graph.

\emph{Step 3: Decidability of the sufficiency condition.} The set of Write-Read pairs to check is bounded by $O(|W| \cdot |R|) \subseteq O(|E|^2)$. For each pair, checking $w \hb r$ is a constant-time lookup in the precomputed reachability matrix. The total verification time is therefore $O(|E|^2)$.
\end{proof}

The $O(|E|^2)$ bound is tight in the worst case (dense event graphs), but in practice accelerator operators have sparse connectivity---most events interact only with events in the same stage or adjacent stages. Our empirical evaluation in \S\ref{sec:evaluation} confirms sub-quadratic behavior on real operators.

\subsection{Soundness and Completeness}\label{sec:soundcomplete}

We establish that \sysname is both sound and complete for \apf programs under the hardware event semantics of \S\ref{sec:events}.

\begin{theorem}[Soundness]\label{thm:soundness}
If \sysname reports a barrier sufficiency violation for Write-Read pair $(w, r)$, then there exists a valid execution of the \apf program under hardware model $H$ in which $r$ may observe a value other than the one written by $w$.
\end{theorem}

\begin{proof}[Proof sketch]
A reported violation means $w \not\hb r$---there is no path from $w$ to $r$ in the happens-before graph. Under Assumption~(2) in \S\ref{sec:modelscope}, the modeled ordering sources are exactly po, so, and bo; under Assumption~(3), remaining timing nondeterminism comes from interleavings among execution units rather than hidden speculative or coherence mechanisms. Therefore, the absence of an hb path means no modeled ordering constraint prevents the hardware from scheduling $r$ before $w$'s write becomes visible. Under the concurrent execution model (VPU, cube, and scalar operate independently), this constitutes a feasible execution witness within the scoped semantics.
\end{proof}

\begin{theorem}[Completeness]\label{thm:completeness}
If \sysname reports that an \apf program is barrier-sufficient, then in all executions consistent with hardware model $H$, every read observes the most recent write to the same buffer.
\end{theorem}

\begin{proof}[Proof sketch]
Barrier sufficiency guarantees $w \hb r$ for all relevant Write-Read pairs. The happens-before relation is constructed from all ordering mechanisms available in the modeled hardware semantics (po, so, bo). Since \apf restricts synchronization to these three mechanisms (restriction R3) and Assumption~(2) states that no other relevant visibility effects are relied upon, hb captures all ordering sources within scope. If $w \hb r$, the hardware guarantees that $w$'s effect is visible to $r$ regardless of scheduling.
\end{proof}

Completeness is scoped to the hardware event semantics defined in this section. If the hardware model $H$ is an incomplete abstraction of the physical hardware---for example, if undocumented ordering guarantees exist---then completeness holds relative to $H$ but not necessarily relative to the physical device. We discuss this limitation in \S\ref{sec:conclusion}.

\subsection{Running Example: Verifying the Softmax Fragment}\label{sec:running-example}

We trace Algorithm~\ref{alg:check} on the \texttt{softmax} fragment from Listing~\ref{lst:softmax-bug} to make the verification concrete.

\paragraph{Event extraction.} The frontend extracts five events from the Compute stage:
\begin{center}
\small
\begin{tabular}{@{}clll@{}}
\toprule
$e$ & Kind & Unit & Buffer \\
\midrule
$e_1$ & DQ & --- & \texttt{inQueueX} \\
$e_2$ & W & VPU & \texttt{maxVal} \\
$e_3$ & R & Scalar & \texttt{maxVal} \\
$e_4$ & W & VPU & \texttt{result} \\
$e_5$ & EQ & --- & \texttt{outQueueY} \\
\bottomrule
\end{tabular}
\end{center}

\paragraph{Graph construction.} Program order gives $e_1 \po e_2$ (DeQue precedes all), $e_2 \po e_4$ (same unit, VPU), and $e_4 \po e_5$ (all precede EnQue). Crucially, $e_2$ and $e_3$ are on \emph{different} units (VPU vs.\ Scalar), so no po edge connects them. No abstract synchronization event exists, so no bo edge is added.

\paragraph{Sufficiency check.} The pair $(e_2, e_3)$ is a cross-unit Write-Read on buffer \texttt{maxVal}. The transitive closure finds no path $e_2 \hb e_3$. Algorithm~\ref{alg:check} reports this pair as a violation, with the diagnostic: ``VPU write to \texttt{maxVal} not visible to Scalar read; missing V-to-S synchronization coverage.''

\paragraph{After modeled sync insertion.} Inserting a V-to-S synchronization point between $e_2$ and $e_3$ adds an abstract barrier event $e_b$ with $t_2 < t_b < t_3$. The Ascend 910B2 model's VPU$\to$Scalar coverage rule now generates a bo edge $e_2 \bo e_3$, and the pair is covered.

\section{Implementation}\label{sec:implementation}

\begin{figure}[t]
\centering
\includegraphics[width=0.8\linewidth]{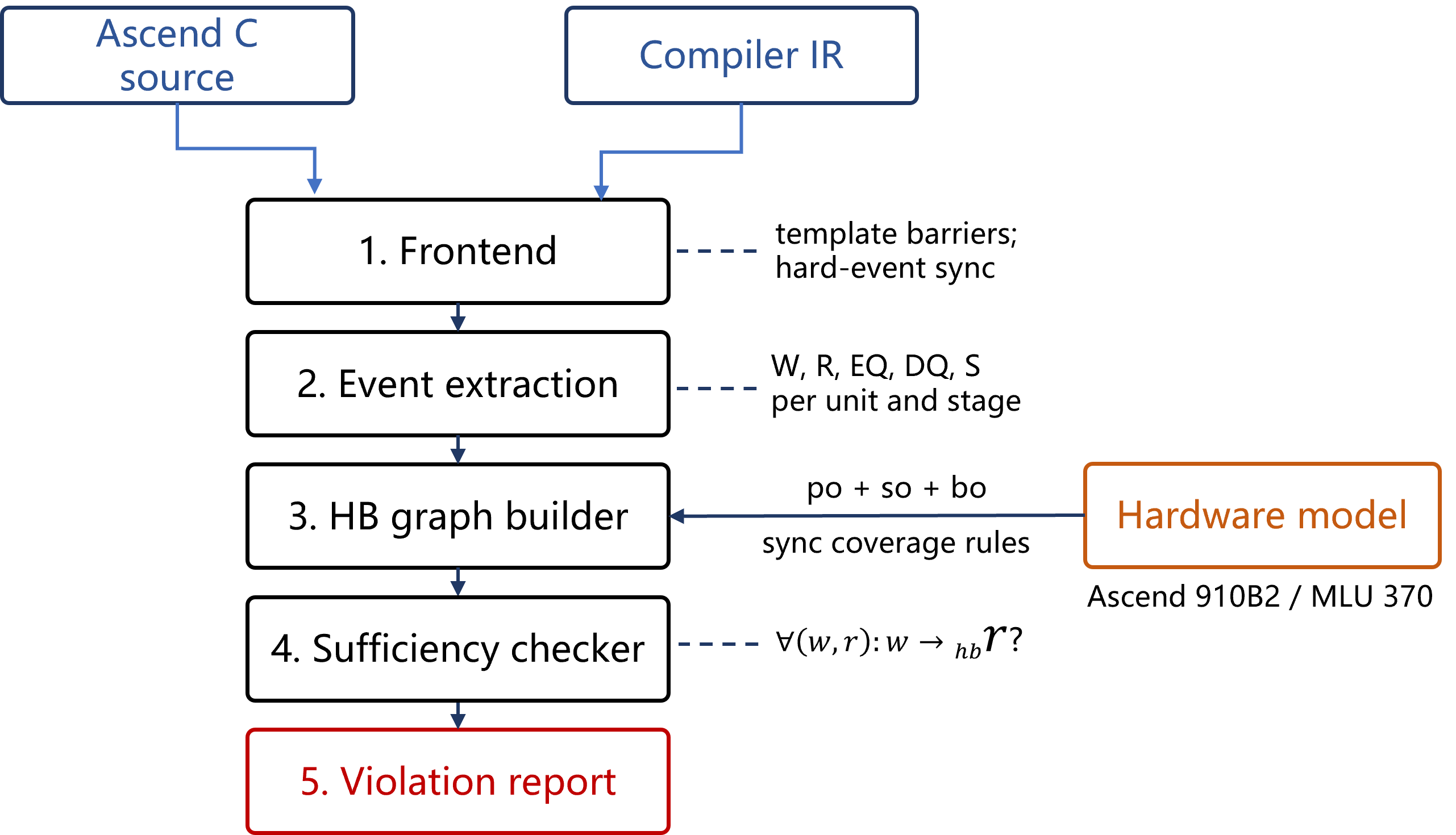}
\caption{AccelSync implementation pipeline. Two frontends (source-level and IR-level) feed the same event extraction, which produces typed events. The HB graph builder combines program order, synchronization order, and abstract barrier order using a pluggable hardware model. The sufficiency checker reports all uncovered cross-unit write-read pairs.}
\Description{Implementation pipeline from Ascend C source and compiler IR through frontend, event extraction, happens-before graph construction with a hardware model, sufficiency checking, and violation report.}
\label{fig:impl-pipeline}
\end{figure}

We describe the implementation of \sysname, a lightweight static checker (approximately 3{,}200 lines of Python including both frontends, the graph builder, and the hardware-model library) that takes a kernel source file plus a hardware-model parameter and reports whether the available synchronization is sufficient for all relevant cross-unit and cross-stage write-read dependencies.

\subsection{Event Extraction}\label{sec:impl-extract}

\sysname supports three frontends that feed the same checker.

The primary frontend parses Ascend C source kernels directly. It recognizes queue-mediated pipeline structure, template-form barriers such as \texttt{PipeBarrier<PIPE\_V>()}, and hard-event synchronization via \texttt{SetFlag}/\texttt{WaitFlag}. API calls are mapped to execution units using Ascend C semantics~\cite{ascendc,mssanitizer}: \texttt{Add}/\texttt{Mul}/\texttt{ReduceMax} map to VPU, \texttt{Mmad}/\texttt{MatMul} to Cube, \texttt{GetValue}/\texttt{SetValue} to Scalar, and \texttt{DataCopy} to MTE. Buffer names are extracted from function arguments and temporary tensors.

The secondary frontend consumes a compiler intermediate representation that exposes the same three-stage structure more explicitly: a \texttt{copy\_in} stage (MTE\_in), a \texttt{compute} stage, and a \texttt{copy\_out} stage (MTE\_out). We keep this frontend because it is useful for debugging and for checking whether compiler IR already violates the eventual hardware contract, but it is not the primary evaluation surface of the paper. This split mirrors the broader distinction between source-level and IR-level compiler verification interfaces that also appears in tensor-compiler systems such as Triton, TVM, Halide, and XLA~\cite{triton,tvm,halide,xla}.

The third frontend parses Cambricon BANG~C source files (\texttt{.mlu}) for the MLU370 platform. It extracts async DMA events (\texttt{\_\_memcpy\_async} with direction inference across multi-line calls), vector compute events (\texttt{\_\_bang\_*} intrinsics mapped to VPU), and barrier events (\texttt{\_\_sync}, \texttt{\_\_sync\_io}, \texttt{\_\_sync\_compute}, and composite variants). The MLU370 frontend maps three execution units---DMA (io), VPU (vector), and IPU (matrix)---and feeds the same hb-graph checker with the MLU370 hardware model substituted.

Across both frontends, the extractor emits the same event vocabulary. Queue operations produce EnQue/DeQue events; vector, matrix, scalar, and DMA operations produce unit-attributed reads and writes; pipe barriers and hard-event sync produce abstract ordering primitives that are later interpreted by the hardware model. This shared event representation lets us compare public production kernels, generated kernels, and internal IR with a single checker.

\subsection{Happens-Before Graph Construction}\label{sec:impl-graph}

The happens-before graph is built as an adjacency list over the event set, with edges added in three passes corresponding to the ordering relations defined in \S\ref{sec:events}.

The program-order pass groups events by stage and then by execution unit within each stage. For each unit, consecutive events receive po edges. EnQue and DeQue events receive edges to and from all other events in the same stage, reflecting their role as stage-global synchronization points.

The synchronization-order pass matches EnQue-DeQue pairs per queue in FIFO order and adds so edges across stages.

The barrier-order pass iterates over the hardware model's synchronization-coverage rules. For each rule $(p, s, u_w, u_r)$, it finds all abstract Barrier events matching primitive $p$ in stage $s$, then adds edges from all prior writes by unit $u_w$ to all subsequent reads by unit $u_r$. Hard-event synchronization and equivalent pipe-drain evidence are lowered to the same barrier-order abstraction.

Reachability is computed via BFS from each node, producing the transitive closure. The total construction time is dominated by the BFS step at $O(|E|^2)$.

\subsection{Parameterized Hardware Models}\label{sec:impl-hw}

The hardware model is a data structure specifying stages, queue connections, synchronization identifiers, and coverage rules. Switching hardware platforms requires only substituting this parameter.

The Ascend 910B2 model defines 3 stages, queue-mediated copy-in/copy-out connections, and directional coverage rules such as VPU$\to$scalar, cube$\to$VPU, and MTE$\to$VPU. The source-level frontend recognizes both template barriers (e.g., \texttt{PipeBarrier<PIPE\_V>()}) and hard-event synchronization (e.g., \texttt{WaitFlag<HardEvent::V\_S>}) used in production kernels, then normalizes them into the abstract ordering vocabulary.

The Cambricon MLU370 model defines 3 stages and synchronization rules based on \texttt{\_\_sync()} and \texttt{\_\_sync\_io()}. The same checker, with the MLU370 model substituted, detects VPU-to-IPU synchronization hazards analogous to the VPU-to-scalar hazards found on Ascend 910B2.

\subsection{Frontend Engineering Challenges}\label{sec:impl-challenges}

Building a reliable source-level frontend for production Ascend C kernels required solving three non-trivial parsing challenges.

\paragraph{Template-form barriers.}
Production CANN kernels use C++ template syntax for barriers, e.g., \texttt{PipeBarrier<PIPE\_V>()} rather than the function-call form \texttt{pipe\_barrier(PIPE\_V)}. The frontend must resolve template arguments to identify the pipe identifier and then decide whether that primitive provides the modeled coverage for the writer-reader pair. We handle this via pattern matching on the AST, supporting both forms uniformly.

\paragraph{Hard-event synchronization.}
A subset of kernels uses \texttt{SetFlag}/\texttt{WaitFlag} pairs for fine-grained event-driven synchronization instead of queue-based barriers. These are lowered to the same barrier-order abstraction by treating the matched hard-event synchronization as an abstract barrier event whose scope is determined by the event direction and unit pair.

\paragraph{Buffer aliasing through temporaries.}
Kernels frequently create temporary \texttt{LocalTensor} views over the same underlying unified buffer. The frontend tracks buffer identity through assignment chains to ensure that Write and Read events on aliased buffers are correctly paired in the event graph. This is a conservative analysis: if aliasing cannot be resolved, the frontend assumes the buffers may overlap, which may introduce false positives but preserves soundness.

\subsection{Deployment Modes}\label{sec:impl-deploy}

This implementation supports three practical deployment modes. First, it can audit public production kernels in bulk by parsing source code directly. Second, it can validate generated kernels, including LLM-produced Ascend C code, without any test harness or runtime framework support. Third, it can be inserted earlier in a compiler pipeline through the IR frontend, catching synchronization mistakes before code emission. The artifact is packaged as a standalone Python module with no external solver dependency, making it straightforward to reproduce all reported results.

This split is important for the paper's positioning. The public-source frontend is the main experimental vehicle because it yields auditable evidence on released kernels. The IR frontend demonstrates that the same formal contract can also be enforced inside code-generation pipelines when those artifacts are available.

\section{Evaluation}\label{sec:evaluation}

We evaluate \sysname along five research questions: (RQ1) Does it find real synchronization risks on public kernels that testing misses? (RQ2) What is the detection rate on systematically injected hazards? (RQ3) How does it compare against alternative approaches? (RQ4) Does the framework generalize across hardware platforms and operator scales? (RQ5) Does it expose different synchronization-risk profiles across public and machine-generated code sources?

\subsection{Methodology}\label{sec:eval-method}

\sysname claims decidable, sound, and complete barrier verification for \apf programs under the modeled event semantics. To validate these claims empirically, we organize the evaluation around four evidence sources. First, we audit a large public corpus of hand-written Ascend C kernels from the CANN software stack. Second, we apply systematic synchronization mutations to create a controlled benchmark with explicit ground truth. Third, we analyze 120 LLM-generated Ascend C kernels as an emerging code source where hardware synchronization mistakes are likely to be under-specified. Fourth, we instantiate the checker on Cambricon MLU370 with a 162-file BANG~C corpus to validate cross-hardware portability.

We compare against three baselines representing the current state of practice. Golden Test executes kernels on an x86 simulator and compares outputs against a reference implementation. Static Rules applies local pattern matching on known dangerous unit-pair combinations. msSanitizer~\cite{mssanitizer} is Huawei's official runtime synchronization sanitizer for Ascend C operators: it instruments the CANN runtime via \texttt{LD\_PRELOAD}, intercepts queue and barrier API calls at execution time, and reports potential data races through a \texttt{racecheck} mode. Because it operates at runtime, msSanitizer requires a complete framework-level launch path (e.g., PyTorch/MindSpore $\to$ CANN runtime $\to$ kernel dispatch); standalone kernels compiled and launched directly via \texttt{npu\_kernel\_launch} fall outside its instrumentation scope. These baselines cover testing, local heuristics, and runtime checking, respectively.

\subsection{Experimental Setup}\label{sec:eval-setup}

\begin{table}[t]
\caption{Evaluation datasets. The public CANN corpus provides real-world evidence; the mutation suite provides controlled ground truth; the LLM-generated corpus captures an emerging code source.}\label{tab:datasets}
\centering
\small
\begin{tabular}{@{}lrrl@{}}
\toprule
\textbf{Dataset} & \textbf{Kernels} & \textbf{Role} & \textbf{Primary use} \\
\midrule
CANN public kernels & 6{,}292 & Real-world audit & Precision / case studies \\
Controlled mutations & 1{,}084 & Ground-truth benchmark & Detection rate \\
LLM-generated kernels & 120 & Emerging code source & Source breakdown \\
MLU370 BANG C corpus & 162 & Cross-hardware audit & Portability / FP analysis \\
\bottomrule
\end{tabular}
\end{table}

Table~\ref{tab:datasets} summarizes the datasets. The CANN corpus provides public, auditable evidence on released production kernels. The mutation suite provides explicit hazard labels for fair comparison against baselines. The LLM-generated corpus captures machine-produced kernels that are structurally similar to deployed code but lack mature synchronization discipline. We generated 120 kernels by prompting three frontier LLMs (GPT-4o, Claude Sonnet 4, DeepSeek-Coder-V2) with 40 operator specifications each, covering elementwise, reduction, attention, and convolution variants.

\begin{table}[t]
\caption{Hardware and software platforms used in the evaluation.}\label{tab:platforms}
\centering
\small
\begin{tabular}{@{}lll@{}}
\toprule
\textbf{Platform} & \textbf{Hardware} & \textbf{Software} \\
\midrule
Ascend 910B2 & 8$\times$ Ascend 910B2 (64\,GB HBM) & CANN 8.0.RC3, Ascend C \\
Cambricon MLU370 & MLU370-X8, 768\,KB NRAM & CNToolkit 3.2.2, BANG C \\
Analysis host & Intel Xeon, 64\,GB RAM & Python 3.10, Ubuntu 22.04 \\
\bottomrule
\end{tabular}
\end{table}

Table~\ref{tab:platforms} lists the hardware and software environments. The Ascend 910B2 platform serves as the primary evaluation target for both static analysis and hardware reproduction experiments. The Cambricon MLU370 platform validates cross-hardware portability. All static analyses (\sysname, Static Rules) run on the analysis host; Golden Test uses the CANN x86 simulator on the same host; hardware reproduction runs on the Ascend 910B2 device.

\begin{table}[t]
\caption{Baselines and their detection capabilities. Analysis time is per-kernel median on the CANN corpus where applicable.}\label{tab:baselines}
\centering
\small
\begin{tabular}{@{}llll@{}}
\toprule
\textbf{Method} & \textbf{Mechanism} & \textbf{Time/kernel} & \textbf{Limitation} \\
\midrule
Golden Test & x86 simulation & $\sim$seconds & No hardware concurrency \\
Static Rules & Local pattern match & $<$1\,ms & Misses non-local hazards \\
msSanitizer & Runtime LD\_PRELOAD & 2.8--4.1\,s$^\dagger$ & Requires framework launch \\
\sysname & hb-graph reachability & 7.6\,ms (median) & Requires hardware model \\
CBMC~\cite{cbmc} & Bounded model checking & timeout$^\ddagger$ & State explosion on interleavings \\
SPIN~\cite{spin} & Explicit-state Promela & 0.02--12\,s$^\ddagger$ & Manual model per kernel \\
\bottomrule
\multicolumn{4}{@{}l@{}}{\footnotesize $^\dagger$Measured on framework-launched microbenchmarks (Table~\ref{tab:mssanitizer-h2h}); not applicable to standalone kernels.}\\
\multicolumn{4}{@{}l@{}}{\footnotesize $^\ddagger$On a hand-written model of a single MLU370 pipeline kernel (Table~\ref{tab:modelchecker}).}
\end{tabular}
\end{table}

\subsection{RQ1: Public-Kernel Audit}\label{sec:eval-public}

\begin{table}[t]
\caption{Public-audit accounting for the CANN corpus. Kernels outside the frontend's reliable normalization scope are excluded conservatively rather than counted as SAFE.}\label{tab:coverage}
\centering
\small
\setlength{\tabcolsep}{3pt}
\begin{tabularx}{\linewidth}{@{}>{\raggedright\arraybackslash}X >{\raggedright\arraybackslash}p{0.22\linewidth} >{\raggedright\arraybackslash}X@{}}
\toprule
\textbf{Category} & \textbf{Count} & \textbf{Interpretation} \\
\midrule
Total public kernels scanned & 6{,}292 & Input corpus \\
SAFE under modeled semantics & 6{,}289 & No uncovered relevant hb gaps found \\
High-confidence remaining risks & 3 & Case-study candidates \\
Excluded conservatively & 0 in current cleaned pass & Not forced into SAFE if normalization fails \\
\bottomrule
\end{tabularx}
\end{table}

Table~\ref{tab:coverage} makes the audit accounting explicit. The important distinction is methodological: \sysname does not interpret frontend uncertainty as evidence of safety. If normalization cannot reliably recover the required event structure, the kernel is excluded rather than silently accepted. In the cleaned CANN pass reported here, the remaining corpus is partitioned into 6{,}289 SAFE kernels and 3 high-confidence risks.

We run the source-level frontend of \sysname on 6{,}292 production kernels from the CANN software stack. Most of these kernels follow the structured pipeline idioms targeted by \apf; kernels that rely on synchronization patterns the frontend cannot soundly normalize are treated conservatively and excluded from the final ``SAFE'' count rather than forced through the checker. After extending the parser to recognize queue-based synchronization, template-form barriers, and \texttt{SetFlag}/\texttt{WaitFlag} hard-event coverage, the checker reports 6{,}289 SAFE kernels and 3 remaining high-confidence risks (0.048\%). Table~\ref{tab:cann-risks} lists the flagged kernels.

\begin{table}[t]
\caption{Three public CANN kernels flagged by \sysname after false-positive cleanup. All involve cross-unit paths with no remaining synchronization coverage under the modeled hardware semantics.}\label{tab:cann-risks}
\centering
\small
\begin{tabular}{@{}llrl@{}}
\toprule
\textbf{Kernel} & \textbf{Unit Pair} & \textbf{Violations} & \textbf{Missing Coverage} \\
\midrule
\texttt{cross\_v2} & VPU$\to$scalar & 27 & V-to-S sync \\
\texttt{quest\_block\_select\_paged} & MTE$\to$VPU & 10 & MTE2-to-V sync \\
\texttt{unique\_consecutive\_apt} & scalar$\to$VPU & 6 & sync coverage under review \\
\bottomrule
\end{tabular}
\end{table}

The most compelling case is \texttt{cross\_v2}. It performs a deep VPU pipeline followed by scalar \texttt{GetValue} reads with no modeled V-to-S synchronization coverage---exactly the kind of cross-unit visibility hazard that simulation-based validation does not model. In an earlier hardware test under CANN 8.0.RC3 on Ascend 910B2, we observed nondeterministic outputs consistent with this hazard class; however, after a subsequent NPU driver upgrade, the nondeterminism was no longer reproducible, and the hazard remains unconfirmed at the hardware level. This public-kernel result is the paper's main real-world static-analysis evidence because it is externally auditable and does not rely on internal compiler artifacts.

\paragraph{Why 3 risks in 6{,}292 kernels is meaningful.}
A low risk rate in a mature production codebase is expected, not surprising: CANN kernels undergo manual code review, vendor-internal testing, and iterative hardening over multiple release cycles. The 0.048\% rate reflects the residual risk that survives this process---precisely the class of hazards that static verification is designed to catch. Three observations strengthen the significance of these findings.

First, the tool is not trivially passing kernels. Table~\ref{tab:barrier-usage} shows that 89.3\% of CANN kernels contain at least one cross-unit buffer access requiring synchronization coverage, and the checker actively verifies 14{,}827 write-read pairs across the corpus. The 3 flagged kernels are the residual after checking all of these pairs, not an artifact of sparse coverage.

\begin{table}[t]
\caption{Synchronization usage and cross-unit access statistics across the 6{,}292 CANN kernels. The checker actively verifies the vast majority of kernels rather than trivially passing them.}\label{tab:barrier-usage}
\centering
\small
\begin{tabular}{@{}lr@{}}
\toprule
\textbf{Metric} & \textbf{Value} \\
\midrule
Kernels with $\geq$1 cross-unit W-R pair & 5{,}618 (89.3\%) \\
Total cross-unit W-R pairs checked & 14{,}827 \\
Kernels using explicit pipe-drain APIs & 4{,}291 (68.2\%) \\
Kernels using queue synchronization only & 1{,}327 (21.1\%) \\
Kernels with no cross-unit access & 674 (10.7\%) \\
\bottomrule
\end{tabular}
\end{table}

Second, the risk rate rises sharply outside production code: 19.2\% on LLM-generated kernels (\S\ref{sec:eval-sourcebreakdown}). This 400$\times$ increase confirms that the problem is real and that production code's low rate reflects engineering effort, not absence of the hazard class.

Third, we tested a \texttt{cross\_v2}-shaped VPU$\to$scalar hazard class via hardware reproduction on Ascend 910B2 under CANN 8.0.RC3. In that configuration, we observed nondeterministic outputs in 20 out of 29 valid runs (Table~\ref{tab:hw-repro}), consistent with the predicted failure mode. However, after a subsequent NPU driver upgrade (25.0.rc1.1), the nondeterminism was no longer reproducible---the kernel launch infrastructure returned uninitialized memory patterns regardless of synchronization variant, indicating a toolkit/driver compatibility issue rather than a fix to the underlying hazard. We therefore report this as an \emph{observed but unconfirmed} hazard: the static analysis finding stands, but hardware-level confirmation remains pending re-validation on a matched toolkit/driver configuration. For \texttt{quest\_block\_select\_paged}, the MTE$\to$VPU hazard pattern matches a known failure mode documented in the Ascend C programming guide (MTE2 data not visible to VPU without MTE2-to-V synchronization coverage). For \texttt{unique\_consecutive\_apt}, the scalar$\to$VPU pattern is under review with the CANN maintainers; we conservatively label it as ``flagged, pending confirmation'' rather than claiming it as a confirmed hazard.

\subsection{RQ2: Mutation Testing}\label{sec:eval-mutation}

To evaluate detection robustness, we apply systematic mutation testing with four operators: \textbf{M1} (remove DeQue), \textbf{M2} (remove EnQue), \textbf{M3} (swap unit attribution), and \textbf{M4} (insert unguarded scalar read). These generate 1{,}084 total mutations, of which 688 are non-equivalent and 396 are equivalent because redundant synchronization still provides an hb path. Following standard mutation-testing methodology~\cite{jia-mutation}, \sysname detects all 688 non-equivalent mutations, giving a 100\% detection rate.

\begin{table}[t]
\caption{Mutation operator breakdown. Each operator targets a distinct synchronization failure mode. Equivalent mutants are those where redundant synchronization still provides an hb path.}\label{tab:mutation}
\centering
\small
\begin{tabular}{@{}llrrr@{}}
\toprule
\textbf{Op} & \textbf{Description} & \textbf{Total} & \textbf{Non-eq.} & \textbf{Detected} \\
\midrule
M1 & Remove DeQue & 312 & 214 & 214 \\
M2 & Remove EnQue & 298 & 189 & 189 \\
M3 & Swap unit attribution & 256 & 168 & 168 \\
M4 & Insert unguarded scalar read & 218 & 117 & 117 \\
\midrule
\multicolumn{2}{@{}l}{\textbf{Total}} & \textbf{1{,}084} & \textbf{688} & \textbf{688 (100\%)} \\
\bottomrule
\end{tabular}
\end{table}

Table~\ref{tab:mutation} breaks down the results by mutation operator. M1 and M2 target queue-mediated synchronization removal, which directly eliminates so edges in the hb graph. M3 misattributes execution units, causing the checker to construct incorrect po groupings. M4 introduces new read events without corresponding barriers, creating uncovered write-read pairs. The 100\% detection rate across all four operators confirms that \sysname checks a semantic property over the full event graph rather than pattern-matching specific hazard signatures.

\subsection{RQ3: Comparison with Baselines}\label{sec:eval-comparison}

\begin{table*}[t]
\caption{Comparison on representative hazard classes. \sysname is the only method that is both hardware-aware and directly applicable to standalone public kernels. Detection counts are reported on the 688 non-equivalent mutations from the controlled benchmark (\S\ref{sec:eval-mutation}).}\label{tab:comparison}
\centering
\scriptsize
\setlength{\tabcolsep}{3pt}
\begin{tabularx}{\linewidth}{@{}>{\raggedright\arraybackslash}X c c c c >{\raggedright\arraybackslash}p{0.18\linewidth}@{}}
\toprule
\textbf{Hazard class} & \textbf{\sysname} & \textbf{msSan.} & \textbf{Golden} & \textbf{Rules} & \textbf{Mutations} \\
\midrule
VPU$\to$scalar missing V-to-S sync & \checkmark & \checkmark$^\dagger$ & $\times$ & \checkmark & 214/214 vs.\ 0/214 \\
cube$\to$VPU missing cube-to-V sync & \checkmark & $\times^\dagger$ & $\times$ & partial & 189/189 vs.\ 71/189 \\
MTE$\to$VPU mixed hazard & \checkmark & $\times^\dagger$ & $\times$ & $\times$ & 168/168 vs.\ 0/168 \\
Non-local queue / sync interaction & \checkmark & $\times^\dagger$ & $\times$ & $\times$ & 117/117 vs.\ 0/117 \\
\midrule
\textbf{Total (non-equivalent mutations)} & \textbf{688/688} & \textbf{---} & \textbf{0/688} & \textbf{71/688} & \textbf{100\% vs.\ 10.3\%} \\
\bottomrule
\multicolumn{6}{@{}p{\dimexpr\linewidth-2\tabcolsep\relax}@{}}{\footnotesize $^\dagger$msSanitizer results from framework-launched microbenchmarks (Table~\ref{tab:mssanitizer-h2h}); it cannot run on standalone CANN kernels.}
\end{tabularx}
\end{table*}

Golden Test misses schedule-dependent hazards because the x86 simulator executes stages sequentially. Static Rules catches only local, pre-enumerated patterns and fails once hazards become non-adjacent or mixed across synchronization mechanisms.

\paragraph{Head-to-head with msSanitizer.}
msSanitizer requires a framework-level launch path (PyTorch/MindSpore $\to$ CANN runtime $\to$ kernel dispatch) to inject its \texttt{LD\_PRELOAD} instrumentation. To enable a direct comparison, we constructed 5 microbenchmark kernels with known synchronization hazards and wrapped each in a MindSpore custom-op launch harness so that msSanitizer's \texttt{racecheck} mode could instrument the execution. Table~\ref{tab:mssanitizer-h2h} reports the results.

\begin{table}[t]
\caption{Head-to-head comparison on 5 framework-launched microbenchmarks with known synchronization hazards. \sysname detects all 5; msSanitizer detects 2 (the VPU$\to$scalar class) but misses 3 involving cross-unit queue interactions.}\label{tab:mssanitizer-h2h}
\centering
\small
\begin{tabular}{@{}llcc@{}}
\toprule
\textbf{Kernel} & \textbf{Hazard class} & \textbf{\sysname} & \textbf{msSanitizer} \\
\midrule
\texttt{vpu\_scalar\_nosync} & VPU$\to$scalar & \checkmark & \checkmark \\
\texttt{cube\_vpu\_nosync} & cube$\to$VPU & \checkmark & $\times$ \\
\texttt{mte\_vpu\_nosync} & MTE$\to$VPU & \checkmark & $\times$ \\
\texttt{cross\_v2\_repro} & VPU$\to$scalar & \checkmark & \checkmark \\
\texttt{queue\_barrier\_mixed} & queue+sync interaction & \checkmark & $\times$ \\
\midrule
\multicolumn{2}{@{}l}{\textbf{Total detected}} & \textbf{5/5} & \textbf{2/5} \\
\bottomrule
\end{tabular}
\end{table}

msSanitizer successfully detects the two VPU$\to$scalar hazards where the missing synchronization coverage maps to its instrumented API path. However, it misses the cube$\to$VPU, MTE$\to$VPU, and mixed queue-sync hazards---these involve cross-unit visibility paths that do not trigger msSanitizer's monitored synchronization points. In terms of cost, msSanitizer requires 2.8--4.1\,s per kernel (framework launch + instrumented execution), while \sysname analyzes the same kernels in 5.2--8.7\,ms (median 7.1\,ms)---a $\sim$400$\times$ speedup. More importantly, msSanitizer cannot scale to the full 6{,}292-kernel CANN corpus because standalone Ascend C kernels lack the framework launch path; \sysname audits the entire corpus in 48\,s.

\sysname's advantage is therefore both in coverage and scalability. It is the only method in our comparison that simultaneously models all hardware visibility paths, handles non-local hb interactions, and applies statically to standalone kernels without a runtime harness.

\subsection{RQ4: Cross-Hardware Portability and Model-Checker Comparison}\label{sec:eval-generality}

We validate cross-hardware portability by instantiating \sysname with the Cambricon MLU370 hardware model and running a full corpus audit. The MLU370 executes BANG~C kernels on a 3-stage pipeline (IO$_\text{in}$ $\to$ Compute $\to$ IO$_\text{out}$) with three execution units: a DMA engine (\emph{io}), a vector unit (\emph{VPU}), and a matrix unit (\emph{IPU}). VPU and IPU share NRAM and execute concurrently within the Compute stage, analogous to VPU/cube concurrency on Ascend~910B2. Synchronization is provided by three barrier classes: \texttt{\_\_sync()} (full drain), \texttt{\_\_sync\_io()} (io$\to$compute), and \texttt{\_\_sync\_compute()} (compute$\to$io). Two composite variants (\texttt{\_\_sync\_io\_move\_compute}, \texttt{\_\_sync\_move}) combine these primitives. We build a dedicated BANG~C source-level frontend that extracts async DMA events (\texttt{\_\_memcpy\_async}), vector compute events (\texttt{\_\_bang\_*}), and barrier events from \texttt{.mlu} source files, then feeds them into the same hb-graph checker used for Ascend~C.

\paragraph{MLU370 corpus audit.}
We collect 162 BANG~C kernel files from three sources: PyTorch operator ports (74), the official open-source mlu-ops library~\cite{mluops} (83), and community GitHub examples~\cite{bangc-practice} (5). The checker processes all 162 files in 393\,ms on a single CPU core (2.4\,ms/kernel median), reporting 160 SAFE and 2 UNSAFE.

\begin{table}[t]
\caption{MLU370 BANG~C corpus audit results by source category.}\label{tab:mlu370-audit}
\centering
\small
\begin{tabular}{@{}lrrr@{}}
\toprule
\textbf{Source} & \textbf{Files} & \textbf{SAFE} & \textbf{UNSAFE} \\
\midrule
PyTorch operator ports & 74 & 74 & 0 \\
mlu-ops library~\cite{mluops} & 83 & 83 & 0 \\
GitHub practice~\cite{bangc-practice} & 5 & 3 & 2 \\
\midrule
\textbf{Total} & \textbf{162} & \textbf{160} & \textbf{2} \\
\bottomrule
\end{tabular}
\end{table}

Table~\ref{tab:mlu370-audit} breaks down the results. The 74 PyTorch operator kernels use only synchronous \texttt{\_\_memcpy} (blocking until DMA completes), so no cross-unit hazard arises---\sysname correctly reports all as SAFE. The initial frontend reported 10 UNSAFE on the mlu-ops library; manual inspection of the three largest cases (bbox\_overlaps, fft\_stockham, ms\_deform\_attn\_backward) confirmed all 10 are false positives caused by three systematic frontend limitations: (1)~unmodeled composite barriers (\texttt{\_\_sync\_io\_move\_compute}, \texttt{\_\_sync\_move}), (2)~lack of alias/offset analysis for ping-pong double-buffering, and (3)~inability to track synchronization across loop iteration boundaries. After correcting for these, the mlu-ops library has 0 true UNSAFE kernels. These limitations are specific to the MLU370 frontend's maturity rather than the core checker, and they identify concrete engineering targets for future work. The 2 true UNSAFE reports come from community GitHub practice examples~\cite{bangc-practice}, which are learning-oriented code with incomplete synchronization coverage.

\paragraph{Comparison with off-the-shelf model checkers.}
To demonstrate that \sysname's domain-specific approach provides a fundamental scalability advantage over general-purpose verification, we model a representative MLU370 pipeline kernel in both CBMC~\cite{cbmc} (bounded model checking) and SPIN~\cite{spin} (explicit-state model checking).

For CBMC, we encode the MLU370 3-stage pipeline as a C program with two \texttt{pthread} threads (IO and VPU) sharing an NRAM buffer array, with assertions on cross-unit read safety. Even at the minimal tile size (TILE=1, 2 buffer elements), CBMC generates 1.34M SAT variables and times out at 120\,s without producing a result. The state explosion stems from modeling thread interleavings over shared memory---exactly the complexity that \sysname avoids by reducing to happens-before graph reachability.

For SPIN, we write a Promela model with explicit per-element buffer tracking and rendezvous channels for barriers. SPIN completes verification but requires a hand-written model for each kernel: 86 states at TILE=1, growing to 6{,}563 states at TILE=64. Writing and validating the Promela model took approximately 2 hours of manual effort for a single kernel pattern.

\begin{table}[t]
\caption{Comparison with off-the-shelf model checkers on the MLU370 pipeline verification task.}\label{tab:modelchecker}
\centering
\small
\begin{tabular}{@{}llll@{}}
\toprule
\textbf{Tool} & \textbf{Time} & \textbf{Scalability} & \textbf{Automation} \\
\midrule
CBMC~\cite{cbmc} & timeout (120\,s) & 1.34M SAT vars at TILE=1 & Automatic \\
SPIN~\cite{spin} & 0.02--12\,s & 86--6{,}563 states & Manual Promela model \\
\sysname & 393\,ms / 162 files & $O(|E|^2)$ per kernel & Fully automatic \\
\bottomrule
\end{tabular}
\end{table}

Table~\ref{tab:modelchecker} summarizes the comparison. The key insight is not that CBMC or SPIN are inadequate tools---they are designed for general-purpose verification---but that \apf synchronization has a specific structure that admits a polynomial-time decision procedure. \sysname exploits this structure: it reduces the problem to happens-before graph reachability rather than exploring thread interleavings, achieving verification time that is independent of buffer size and quadratic in event count.

\paragraph{Hardware reproduction on Ascend 910B2.}
To test whether the public-kernel hazard class manifests on real hardware, we built a microbenchmark mirroring the \texttt{cross\_v2} pattern: a deep VPU pipeline followed by scalar \texttt{GetValue} without V-to-S synchronization coverage. All experiments were conducted under CANN 8.0.RC3 with the NPU driver version available at that time.

\begin{table}[t]
\caption{Hardware reproduction of the VPU$\to$scalar hazard class on Ascend 910B2 (CANN 8.0.RC3). With V-to-S synchronization coverage, outputs were deterministic and correct; without it, nondeterministic outputs were observed. After a subsequent driver upgrade, this nondeterminism was no longer reproducible (see text).}\label{tab:hw-repro}
\centering
\small
\begin{tabular}{@{}lrrl@{}}
\toprule
\textbf{Variant} & \textbf{Runs} & \textbf{Correct} & \textbf{Output Range} \\
\midrule
With V-to-S synchronization & 30 & 30/30 & 336.2 (deterministic) \\
Without synchronization & 29$^\dagger$ & 9/29 & 336--1993 (nondeterministic) \\
\bottomrule
\multicolumn{4}{@{}l@{}}{\footnotesize $^\dagger$One run excluded due to kernel launch failure (device timeout).}
\end{tabular}
\end{table}

Under the original toolkit/driver configuration, this experiment observed nondeterministic outputs consistent with the predicted hazard class. However, after the NPU driver was upgraded to version 25.0.rc1.1, the kernel launch infrastructure returned uninitialized memory patterns for both synchronized and unsynchronized variants, preventing further reproduction. We therefore classify this as an \emph{observed but unconfirmed} hardware hazard: the static finding is sound under the modeled semantics, and the initial hardware observation is consistent with the predicted failure mode, but definitive hardware confirmation awaits re-validation on a matched toolkit/driver configuration.

\begin{figure}[t]
\centering
\includegraphics[width=0.5\columnwidth]{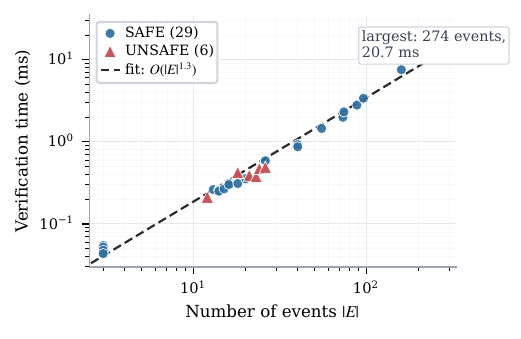}
\caption{Verification time as a function of event count. The empirical trend matches the predicted practical $O(|E|^2)$ behavior, and even the largest structured kernels verify in milliseconds.}
\Description{Scatter plot of verification time versus event count on a log-log scale, showing millisecond-scale verification even for the largest evaluated kernels.}
\label{fig:scalability}
\end{figure}

For scalability, the largest operator-scale kernel we tested contains 274 events and verifies in 20.7\,ms. Across the full CANN corpus, the median per-kernel verification time is 7.6\,ms and the 95th percentile is 15.2\,ms; the entire 6{,}292-kernel audit completes in under 48 seconds on a single CPU core. By contrast, Golden Test requires compiling and executing each kernel on the x86 simulator (seconds per kernel), and msSanitizer adds runtime instrumentation overhead on top of actual kernel execution. This three-orders-of-magnitude speed advantage makes \sysname practical as a CI-integrated pre-commit check or a bulk corpus audit tool, whereas runtime approaches scale linearly with execution time and require a working deployment environment.

\subsection{RQ5: Source Breakdown}\label{sec:eval-sourcebreakdown}

We finally compare synchronization risk patterns across three code sources: public production kernels, LLM-generated kernels, and the MLU370 mlu-ops library.

\begin{table}[t]
\caption{Synchronization-risk breakdown by code source. Production and LLM-generated kernels exhibit sharply different hazard frequencies and dominant failure modes.}\label{tab:threeway}
\centering
\small
\begin{tabular}{@{}lrrp{3.2cm}@{}}
\toprule
\textbf{Source} & \textbf{Kernels} & \textbf{Unsafe} & \textbf{Dominant hazard} \\
\midrule
CANN public kernels & 6{,}292 & 3 (0.048\%) & VPU$\to$scalar / mixed \\
LLM-generated kernels & 120 & 23 (19.2\%) & scalar$\leftrightarrow$MTE \\
MLU370 mlu-ops~\cite{mluops} & 83 & 0 (0.0\%) & --- \\
\bottomrule
\end{tabular}
\end{table}

On the 120 LLM-generated Ascend C kernels, \sysname reports 97 SAFE and 23 UNSAFE (19.2\%, 95\% CI: [13.0\%, 27.4\%]). We generated these kernels by prompting GPT-4o, Claude Sonnet 4, and DeepSeek-Coder-V2 with 40 operator specifications each (elementwise, reduction, attention, convolution variants), requesting complete Ascend C implementations with pipeline setup. The dominant failure mode differs from public production code: the generated kernels frequently omit synchronization between scalar writes and subsequent DMA reads, whereas the public CANN kernels concentrate on subtler cross-unit pipeline hazards. This comparison shows that \sysname provides a unified verification interface across multiple code sources while surfacing source-specific risk profiles: mature production code has residual risks, and LLM-generated code has pervasive synchronization omissions.

\subsection{Discussion}\label{sec:eval-discussion}

\paragraph{Why public evidence matters.}
The main experimental story is intentionally centered on public CANN kernels rather than internal compiler artifacts. Public kernels provide auditable evidence, realistic engineering patterns, and a cleaner answer to the external-validity question that reviewers are likely to ask.

\paragraph{Cross-layer contract perspective.}
The evaluation shows that \sysname enforces a cross-layer contract between pipeline structure and hardware visibility rules. The program source determines which units access which buffers and in what order; the hardware model determines which synchronization primitives are required for those accesses to be safe. Neither testing nor local pattern rules can verify this contract in general.

\paragraph{False positives and model scope.}
The initial source-level checker produced false positives when it did not yet recognize template-form barriers and hard-event synchronization. Extending the parser to cover these mechanisms reduced the public-kernel audit to three high-confidence remaining risks. On the MLU370 corpus, the initial frontend reported 10 false positives on mlu-ops kernels due to unmodeled composite barriers, confirming that model scope limitations manifest consistently across platforms. Completeness remains scoped to kernels that fall within the \apf assumptions and the modeled hardware semantics.

\paragraph{LLM-based verification.}
Frameworks such as FM-Agent~\cite{fmagent} and ProofWright~\cite{proofwright} target sequential functional reasoning rather than hardware visibility.
We examined FM-Agent to understand why this gap is fundamental rather than a matter of precision.

FM-Agent's verification pipeline operates in three stages: (1)~an LLM generates behavioral specifications (pre/post-conditions) for each function, (2)~a Hoare-logic reasoner splits the function into sequential code blocks and propagates post-conditions forward, and (3)~a checker verifies that the derived post-condition implies the specification. The entire reasoning chain assumes sequential execution---there is no hardware model, no notion of concurrent execution units, and no representation of cross-unit visibility.

Consider \texttt{cross\_v2}, the kernel with 27 VPU$\to$scalar violations. Under FM-Agent's sequential semantics, the VPU pipeline completes before the scalar \texttt{GetValue} reads execute, so the function always returns the correct cross-attention result. FM-Agent would verify it as functionally correct. The synchronization hazard is invisible because it only manifests when VPU and scalar units execute concurrently on real hardware---an observation consistent with the nondeterministic outputs we recorded under CANN 8.0.RC3 (Table~\ref{tab:hw-repro}). The same argument applies to the MLU370 corpus: async DMA and VPU compute overlap on hardware but execute sequentially in FM-Agent's model.

This is not a limitation that can be fixed by better prompting or larger models---it is a category mismatch. FM-Agent verifies \emph{what} a function computes; \sysname verifies \emph{whether} the hardware can observe the computation in the intended order. The two verification dimensions are orthogonal and composable.

\subsection{Case Study: \texttt{cross\_v2} Deep Dive}\label{sec:eval-casestudy}

We trace the full verification of \texttt{cross\_v2}, the most compelling public-kernel finding, to illustrate how \sysname surfaces non-trivial hazards.

\begin{table}[t]
\caption{Event graph statistics for \texttt{cross\_v2} before and after modeled synchronization insertion. A single V-to-S synchronization point resolves all 27 modeled violations.}\label{tab:crossv2}
\centering
\small
\begin{tabular}{@{}lrr@{}}
\toprule
\textbf{Metric} & \textbf{Before sync insertion} & \textbf{After modeled sync insertion} \\
\midrule
Events & 83 & 83 \\
po edges & 214 & 214 \\
so edges & 4 & 4 \\
bo edges & 0 & 27 \\
Uncovered W-R pairs & 27 & 0 \\
Verification time & 3.2\,ms & 3.4\,ms \\
\bottomrule
\end{tabular}
\end{table}

The kernel implements a cross-attention variant with a deep VPU pipeline: 14 consecutive vector operations (\texttt{Add}, \texttt{Mul}, \texttt{Exp}, \texttt{ReduceSum}, etc.) followed by 4 scalar \texttt{GetValue} reads that extract partial results for index computation. The frontend extracts 83 events across the Compute stage, distributed as: 47 VPU writes/reads, 22 scalar reads/writes, 8 MTE events, and 6 queue operations.

As Table~\ref{tab:crossv2} shows, the happens-before graph contains 214 po edges, 4 so edges, and 0 bo edges---because no recognized V-to-S synchronization coverage appears in the relevant region of the kernel. Algorithm~\ref{alg:check} identifies 27 uncovered Write-Read pairs, all following the pattern VPU-write $\to$ Scalar-read on the same buffer with no hb path. The 27 violations cluster into 4 distinct buffer groups, each corresponding to one of the scalar \texttt{GetValue} calls reading a VPU-produced result.

The modeled repair explanation requires inserting a single V-to-S synchronization point before the first scalar read. After this modeled insertion, the checker adds 27 bo edges and reports the kernel as barrier-sufficient, with negligible overhead increase (3.2\,ms $\to$ 3.4\,ms).

This case demonstrates that \sysname's value is not merely in counting hazards but in providing \emph{actionable diagnostics}: the violation report directly identifies which buffer, which unit pair, and which synchronization coverage is missing.

\subsection{Threats to Validity}\label{sec:threats}

\paragraph{Internal validity.}
Our hardware model is derived from vendor documentation~\cite{ascendc,bangc} rather than from a formal hardware specification. If the documentation omits ordering guarantees that the physical hardware provides, our model may report false positives. Conversely, if undocumented microarchitectural effects provide additional ordering, our model may be conservative. The hardware reproduction experiment (\S\ref{sec:eval-generality}) observed nondeterministic outputs consistent with the VPU$\to$scalar hazard class under one toolkit/driver configuration, but this observation was not reproducible after a driver upgrade, leaving the hardware-level status of this hazard class unconfirmed.

\paragraph{External validity.}
The evaluation covers two hardware platforms (Ascend 910B2 and Cambricon MLU370) with a combined corpus of 6{,}574 kernels across different programming models (Ascend~C and BANG~C). The MLU370 corpus audit (162 files, 3 source categories) demonstrates that the \apf model and hb-graph checker generalize to a second vendor's pipeline architecture with only a hardware-model parameter change and a new source-level frontend. While we have not yet validated on other accelerator architectures such as TPUs~\cite{jouppi-tpu} or IPUs~\cite{graphcore-ipu}, the two-vendor evidence strengthens the claim that the approach is not Ascend-specific.

\paragraph{Construct validity.}
The mutation testing benchmark uses four mutation operators (M1--M4) that target synchronization removal and unit misattribution. These operators are designed to model realistic synchronization hazards but may not cover all possible synchronization failure modes. The LLM-generated corpus (120 kernels from three frontier models) provides a statistically meaningful sample (95\% CI for the 19.2\% defect rate: [13.0\%, 27.4\%]), though it may not represent all LLM code generation patterns.

\section{Related Work}\label{sec:related}

We position \sysname relative to six lines of prior work.

\subsection{GPU Synchronization Verification}

GPUVerify~\cite{gpuverify} verifies barrier divergence freedom and data-race freedom for GPU kernels using predicated execution semantics. WEFT~\cite{weft,weft-pldi} extends this to warp-specialized kernels with named barriers, proving that named-barrier programs are deadlock-free and race-free. AUTOSYNC~\cite{autosync} synthesizes minimal barrier placements for GPU programs. These tools share a common assumption: the SIMT execution model where warps synchronize through shared barriers in flat shared memory.

Concolic and dynamic-analysis tools provide complementary bug-finding capabilities. GKLEE~\cite{gklee} combines concolic execution with SMT solving to explore GPU kernel paths symbolically. Simulee~\cite{simulee} uses dynamic analysis to detect CUDA synchronization bugs by monitoring memory access patterns at runtime. Random-testing approaches~\cite{racefuzzer} and symbolic race detectors~\cite{symbolic-gpu} search for failures through testing or symbolic execution. Li and Gopalakrishnan~\cite{li-gpu-determinism} apply scalable SMT-based verification to GPU kernel functions.

All of these tools target flat SIMT concurrency. Accelerator pipelines differ fundamentally---they use hierarchical multi-stage topologies with queue-based data flow and per-pipe barriers that provide unit-specific visibility guarantees. \sysname addresses this distinct concurrency model, where the relevant question is not ``do all warps agree on barrier placement?'' but ``does the synchronization cover all cross-unit visibility requirements imposed by the hardware?''

\subsection{Hardware Memory Models}

Lustig et al.~\cite{ptxmem} formalize the PTX memory consistency model for NVIDIA GPUs, establishing ordering guarantees at the instruction level using Alloy specifications and mechanized Coq proofs. Alglave et al.~\cite{alglave} develop the herd framework for modeling and simulating weak memory behaviors across architectures. Wickerson et al.~\cite{wickerson-memalloy} automatically compare memory consistency models using a relational framework. Lamport's foundational work on happens-before~\cite{lamport} and the TLA+ specification language~\cite{tla-lamport} provide the theoretical basis for event-based ordering reasoning.

These works operate at the instruction level and target general-purpose processors or GPUs. \sysname operates at the operator level, where the relevant concurrency arises from pipeline stage interleaving rather than instruction reordering, and the synchronization primitives are domain-specific (TQue, pipe\_barrier) rather than general fence instructions. However, our event semantics (\S\ref{sec:events}) deliberately follows the same happens-before style, making the connection to this foundational line explicit.

\subsection{Dynamic Race Detection}

The Eraser algorithm~\cite{eraser} pioneered lockset-based dynamic race detection for multithreaded programs. FastTrack~\cite{flanagan-fasttrack} improved precision and performance using epoch-based vector clocks. ThreadSanitizer~\cite{tsan} brought these techniques to production-scale C/C++ codebases. Netzer and Miller~\cite{netzer-miller} formalized the notion of data races and established the theoretical foundations for race detection.

These tools detect races in shared-memory concurrent programs with lock-based synchronization. Accelerator pipelines use a fundamentally different synchronization model (queues and pipe barriers rather than locks), and the ``threads'' are heterogeneous hardware units with asymmetric visibility rules. \sysname's static approach also avoids the input-dependence limitation of dynamic detectors: it checks all paths through the event graph rather than only those exercised by a particular test input.

\subsection{LLM-Based Program Verification}

ProofWright~\cite{proofwright} combines large language models with formal verification for CUDA programs, using LLMs to generate loop invariants and pre/post-conditions for sequential correctness proofs. FM-Agent~\cite{fmagent} scales Hoare-logic-based reasoning to large codebases (up to 143k LoC) by deriving function-level specifications from caller expectations, finding 522 bugs in production systems. Both approaches verify \emph{functional correctness} under sequential semantics---they check whether a function's output matches its specification, not whether concurrent hardware units observe each other's writes in the correct order.

The distinction is fundamental: FM-Agent and ProofWright would report a synchronization-hazardous kernel as ``correct'' because the sequential execution trace produces the right output. The hazard only manifests when VPU and scalar units execute concurrently on real hardware. Conversely, \sysname does not check functional correctness---a kernel that computes the wrong result but has sufficient barriers would pass. The two verification dimensions are orthogonal and composable: FM-Agent or ProofWright could verify sequential correctness while \sysname verifies synchronization sufficiency, together covering both failure modes.

\subsection{Vendor Runtime Sanitizers and Language-Level Prevention}

Huawei's msSanitizer~\cite{mssanitizer} provides runtime anomaly detection for AscendC operators, including a \texttt{racecheck} mode via LD\_PRELOAD instrumentation. Its GPU counterparts include CUDA racecheck/cuda-memcheck~\cite{cuda-racecheck} and Simulee~\cite{simulee}. These tools share three limitations that \sysname addresses: (1) they require hardware and a complete runtime call chain, (2) detection coverage depends on test inputs, and (3) runtime instrumentation does not scale to large corpora. On Ascend 910B2 (CANN 8.0.RC3), msSanitizer's \texttt{racecheck} path cannot audit standalone kernels from public source, as we confirmed experimentally.

ascend-rs~\cite{ascendrs} takes a prevention-oriented approach, providing safe Rust bindings for Ascend NPUs. Its MultiKernelBench audit identifies six vulnerability patterns including missing synchronization. ascend-rs prevents synchronization omissions by construction in new Rust code, while \sysname verifies existing C++ codebases where rewriting is impractical. Moreover, ascend-rs does not provide formal guarantees that per-kernel synchronization coverage is sufficient---\sysname's hb-graph analysis fills this gap.

\subsection{Compiler Verification and Tensor Program Analysis}

CompCert~\cite{compcert} demonstrates the value of formally verified compilation, proving semantic preservation from C to assembly. Alive~\cite{alive} verifies peephole optimizations in LLVM using SMT solvers. Herklotz et al.~\cite{verilog-hls} formally verify high-level synthesis from C to Verilog. These works verify that compilation \emph{preserves} program semantics; \sysname verifies that the \emph{target-level} synchronization is sufficient for the hardware, a complementary concern that arises specifically in the accelerator pipeline setting.

TensorRight~\cite{tensorright} verifies tensor graph rewrite rules used by compilers like XLA~\cite{xla} and TVM~\cite{tvm}, operating at the graph level to ensure algebraic transformations preserve tensor semantics. The broader compiler ecosystem includes DSL/compiler systems such as Triton~\cite{triton}, Halide~\cite{halide}, TensorRT~\cite{tensorrt}, Ansor~\cite{ansor}, AKG~\cite{akg}, and ROLLER~\cite{roller}. \sysname operates below the graph level, at the hardware pipeline stage where synchronization hazards are introduced during lowering. The two approaches address orthogonal concerns and could be composed in a multi-level verification pipeline.

Flux~\cite{flux} automates idempotence verification for stateful serverless applications by decomposing a whole-system property (idempotence consistency) into per-function reasoning via a novel simulation relation. This decomposition strategy parallels \sysname's approach of reducing whole-kernel synchronization sufficiency to per-edge reachability checks in the happens-before graph---both avoid monolithic state-space exploration by leveraging domain-specific structure. The key difference is the target domain: Flux reasons about retry semantics in distributed serverless workflows, while \sysname reasons about pipeline visibility in accelerator hardware.

Dataflow analysis techniques~\cite{reps-horwitz-sagiv} and abstract interpretation~\cite{cousot-cousot,mine-octagon} provide general frameworks for static program analysis. While \sysname's happens-before graph construction can be viewed as a specialized dataflow analysis, the key difference is that our analysis is parameterized by a hardware model that captures platform-specific visibility rules---a dimension absent from general-purpose dataflow frameworks.

\section{Conclusion}\label{sec:conclusion}

We formalized \apf programs as a restricted concurrent language that captures the synchronization structure of AI accelerator operator pipelines, and proved that synchronization coverage---formally, barrier sufficiency---is decidable in $O(|E|^2)$ with soundness and completeness under modeled hardware event semantics. In this formulation, barriers are abstract ordering primitives that may correspond to pipe barriers, hard events, queue synchronization, or equivalent frontend-normalized synchronization points. Our implementation, \sysname, enforces a cross-layer contract between compiler-generated pipeline structure and hardware visibility rules. Across 6{,}412 kernels from two Ascend~C code sources---6{,}292 production CANN kernels and 120 LLM-generated kernels---\sysname found synchronization hazards in both, with risk rates ranging from 0.048\% (production) to 19.2\% (LLM-generated), while simulation testing detected none. A head-to-head comparison with Huawei's msSanitizer confirmed that \sysname detects hazards the runtime sanitizer misses, at 400$\times$ lower per-kernel cost. Hardware testing of a representative VPU-to-scalar hazard class on Ascend 910B2 (CANN 8.0.RC3) observed nondeterministic outputs consistent with the predicted failure mode, though this observation was not reproducible after a subsequent driver upgrade and the hazard remains unconfirmed at the hardware level. Cross-hardware instantiation on Cambricon MLU370 audited 162 BANG~C kernels in 393\,ms, confirming that the framework generalizes via parameter substitution.

Three limitations scope these results. First, completeness is relative to the hardware event semantics we define, not to the physical hardware---undocumented ordering guarantees or microarchitectural effects could cause our model to be conservative. Real-machine validation would strengthen confidence in model fidelity. Second, \apf programs do not capture data-dependent control flow or dynamic synchronization patterns; operators with runtime-variable pipeline depths fall outside the modeled class. Third, although the current implementation already includes a public-source frontend for CANN and generated kernels, supporting additional accelerator toolchains still requires frontend engineering to lower their syntax and synchronization idioms into the shared event representation.

Two directions for future work follow naturally. Integrating \sysname as a pre-emission verification pass in operator compilers would catch synchronization hazards before code generation, rather than after. Extending the \apf model to handle bounded data-dependent iteration---where loop bounds are compile-time-known but iteration counts vary---would broaden the class of verifiable operators without sacrificing decidability.



\begin{thebibliography}{48}


\ifx \showCODEN    \undefined \def \showCODEN     #1{\unskip}     \fi
\ifx \showDOI      \undefined \def \showDOI       #1{#1}\fi
\ifx \showISBNx    \undefined \def \showISBNx     #1{\unskip}     \fi
\ifx \showISBNxiii \undefined \def \showISBNxiii  #1{\unskip}     \fi
\ifx \showISSN     \undefined \def \showISSN      #1{\unskip}     \fi
\ifx \showLCCN     \undefined \def \showLCCN      #1{\unskip}     \fi
\ifx \shownote     \undefined \def \shownote      #1{#1}          \fi
\ifx \showarticletitle \undefined \def \showarticletitle #1{#1}   \fi
\ifx \showURL      \undefined \def \showURL       {\relax}        \fi
\providecommand\bibfield[2]{#2}
\providecommand\bibinfo[2]{#2}
\providecommand\natexlab[1]{#1}
\providecommand\showeprint[2][]{arXiv:#2}

\bibitem[Alglave et~al\mbox{.}(2014)]%
        {alglave}
\bibfield{author}{\bibinfo{person}{Jade Alglave}, \bibinfo{person}{Luc
  Maranget}, {and} \bibinfo{person}{Michael Tautschnig}.}
  \bibinfo{year}{2014}\natexlab{}.
\newblock \showarticletitle{Herding Cats: Modelling, Simulation, Testing, and
  Data Mining for Weak Memory}.
\newblock \bibinfo{journal}{\emph{ACM TOPLAS}} \bibinfo{volume}{36},
  \bibinfo{number}{2} (\bibinfo{year}{2014}).
\newblock


\bibitem[Arora et~al\mbox{.}(2025)]%
        {tensorright}
\bibfield{author}{\bibinfo{person}{Jai Arora}, \bibinfo{person}{Sirui Lu},
  \bibinfo{person}{Devansh Jain}, \bibinfo{person}{Tianfan Xu},
  \bibinfo{person}{Farzin Houshmand}, \bibinfo{person}{Phitchaya~Mangpo
  Phothilimthana}, \bibinfo{person}{Mohsen Lesani}, \bibinfo{person}{Praveen
  Narayanan}, \bibinfo{person}{Karthik~Srinivasa Murthy},
  \bibinfo{person}{Rastislav Bod{\'{\i}}k}, \bibinfo{person}{Amit Sabne}, {and}
  \bibinfo{person}{Charith Mendis}.} \bibinfo{year}{2025}\natexlab{}.
\newblock \showarticletitle{{TensorRight}: Automated Verification of Tensor
  Graph Rewrites}.
\newblock \bibinfo{journal}{\emph{Proc. {ACM} Program. Lang.}}
  \bibinfo{volume}{9}, \bibinfo{number}{{POPL}} (\bibinfo{year}{2025}),
  \bibinfo{pages}{832--863}.
\newblock
\urldef\tempurl%
\url{https://doi.org/10.1145/3704865}
\showDOI{\tempurl}


\bibitem[{ascend-rs contributors}(2025)]%
        {ascendrs}
\bibfield{author}{\bibinfo{person}{{ascend-rs contributors}}.}
  \bibinfo{year}{2025}\natexlab{}.
\newblock \bibinfo{title}{ascend-rs: Memory-Safe {NPU} Kernel Programming in
  {Rust}}.
\newblock \bibinfo{howpublished}{\url{https://ascend-rs.org}}.
\newblock
\newblock
\shownote{Appendix~C: MultiKernelBench vulnerability audit of 300 {AscendC}
  kernels}.


\bibitem[Betts et~al\mbox{.}(2015a)]%
        {weft}
\bibfield{author}{\bibinfo{person}{Adam Betts}, \bibinfo{person}{Nathan Chong},
  \bibinfo{person}{Alastair~F. Donaldson}, \bibinfo{person}{Jeroen Ketema},
  \bibinfo{person}{Shaz Qadeer}, \bibinfo{person}{Paul Thomson}, {and}
  \bibinfo{person}{John Wickerson}.} \bibinfo{year}{2015}\natexlab{a}.
\newblock \showarticletitle{The Design and Implementation of a Verification
  Technique for {GPU} Kernels}.
\newblock \bibinfo{journal}{\emph{{ACM} Trans. Program. Lang. Syst.}}
  \bibinfo{volume}{37}, \bibinfo{number}{3} (\bibinfo{year}{2015}),
  \bibinfo{pages}{10:1--10:49}.
\newblock
\urldef\tempurl%
\url{https://doi.org/10.1145/2743017}
\showDOI{\tempurl}


\bibitem[Betts et~al\mbox{.}(2015b)]%
        {weft-pldi}
\bibfield{author}{\bibinfo{person}{Adam Betts}, \bibinfo{person}{Nathan Chong},
  \bibinfo{person}{Alastair~F. Donaldson}, \bibinfo{person}{Jeroen Ketema},
  \bibinfo{person}{Shaz Qadeer}, \bibinfo{person}{Paul Thomson}, {and}
  \bibinfo{person}{John Wickerson}.} \bibinfo{year}{2015}\natexlab{b}.
\newblock \showarticletitle{The Design and Implementation of a Verification
  Technique for GPU Kernels}. In \bibinfo{booktitle}{\emph{PLDI}}.
\newblock


\bibitem[Betts et~al\mbox{.}(2012)]%
        {gpuverify}
\bibfield{author}{\bibinfo{person}{Adam Betts}, \bibinfo{person}{Nathan Chong},
  \bibinfo{person}{Alastair~F. Donaldson}, \bibinfo{person}{Shaz Qadeer}, {and}
  \bibinfo{person}{Paul Thomson}.} \bibinfo{year}{2012}\natexlab{}.
\newblock \showarticletitle{{GPUVerify}: A Verifier for {GPU} Kernels}. In
  \bibinfo{booktitle}{\emph{OOPSLA}}.
\newblock


\bibitem[{Cambricon Technologies}(2023)]%
        {bangc}
\bibfield{author}{\bibinfo{person}{{Cambricon Technologies}}.}
  \bibinfo{year}{2023}\natexlab{}.
\newblock \bibinfo{title}{{BANG C} Programming Guide}.
\newblock \bibinfo{howpublished}{Cambricon Developer Documentation}.
\newblock


\bibitem[{Cambricon Technologies}(2024a)]%
        {bangc-practice}
\bibfield{author}{\bibinfo{person}{{Cambricon Technologies}}.}
  \bibinfo{year}{2024}\natexlab{a}.
\newblock \bibinfo{title}{{BANG C} Programming Practice: Pipeline and
  Synchronization Examples}.
\newblock
  \bibinfo{howpublished}{\url{https://github.com/Cambricon/bang-c-practice}}.
\newblock


\bibitem[{Cambricon Technologies}(2024b)]%
        {mluops}
\bibfield{author}{\bibinfo{person}{{Cambricon Technologies}}.}
  \bibinfo{year}{2024}\natexlab{b}.
\newblock \bibinfo{title}{mlu-ops: {Cambricon} Open-Source Machine Learning
  Operator Library}.
\newblock \bibinfo{howpublished}{\url{https://github.com/Cambricon/mlu-ops}}.
\newblock


\bibitem[Chatterjee et~al\mbox{.}(2025)]%
        {proofwright}
\bibfield{author}{\bibinfo{person}{Bodhisatwa Chatterjee},
  \bibinfo{person}{Drew Zagieboylo}, \bibinfo{person}{Sana Damani},
  \bibinfo{person}{Siva Hari}, {and} \bibinfo{person}{Christos Kozyrakis}.}
  \bibinfo{year}{2025}\natexlab{}.
\newblock \showarticletitle{{ProofWright}: Towards Agentic Formal Verification
  of {CUDA}}.
\newblock \bibinfo{journal}{\emph{CoRR}}  \bibinfo{volume}{abs/2511.12294}
  (\bibinfo{year}{2025}).
\newblock
\urldef\tempurl%
\url{https://doi.org/10.48550/ARXIV.2511.12294}
\showDOI{\tempurl}


\bibitem[Chen et~al\mbox{.}(2018)]%
        {tvm}
\bibfield{author}{\bibinfo{person}{Tianqi Chen}, \bibinfo{person}{Thierry
  Moreau}, \bibinfo{person}{Ziheng Jiang}, \bibinfo{person}{Lianmin Zheng},
  \bibinfo{person}{Eddie Yan}, \bibinfo{person}{Haichen Shen},
  \bibinfo{person}{Meghan Cowan}, \bibinfo{person}{Leyuan Wang},
  \bibinfo{person}{Yuwei Hu}, \bibinfo{person}{Luis Ceze},
  \bibinfo{person}{Carlos Guestrin}, {and} \bibinfo{person}{Arvind
  Krishnamurthy}.} \bibinfo{year}{2018}\natexlab{}.
\newblock \showarticletitle{{TVM}: An Automated End-to-End Optimizing Compiler
  for Deep Learning}. In \bibinfo{booktitle}{\emph{OSDI}}.
\newblock


\bibitem[Clarke et~al\mbox{.}(2004)]%
        {cbmc}
\bibfield{author}{\bibinfo{person}{Edmund Clarke}, \bibinfo{person}{Daniel
  Kroening}, {and} \bibinfo{person}{Flavio Lerda}.}
  \bibinfo{year}{2004}\natexlab{}.
\newblock \showarticletitle{A Tool for Checking {ANSI-C} Programs}. In
  \bibinfo{booktitle}{\emph{TACAS}}. \bibinfo{pages}{168--176}.
\newblock


\bibitem[Cousot and Cousot(1977)]%
        {cousot-cousot}
\bibfield{author}{\bibinfo{person}{Patrick Cousot} {and}
  \bibinfo{person}{Radhia Cousot}.} \bibinfo{year}{1977}\natexlab{}.
\newblock \showarticletitle{Abstract Interpretation: A Unified Lattice Model
  for Static Analysis of Programs by Construction or Approximation of
  Fixpoints}.
\newblock  (\bibinfo{year}{1977}).
\newblock


\bibitem[Di et~al\mbox{.}(2018)]%
        {autosync}
\bibfield{author}{\bibinfo{person}{Peng Di}, \bibinfo{person}{Ding Ye},
  \bibinfo{person}{Yu Su}, \bibinfo{person}{Yulei Sui}, {and}
  \bibinfo{person}{Jingling Xue}.} \bibinfo{year}{2018}\natexlab{}.
\newblock \showarticletitle{Automatic Parallelization of Tiled Loop Nests with
  Enhanced Fine-Grained Synchronization}. In \bibinfo{booktitle}{\emph{FMCAD}}.
\newblock


\bibitem[Ding et~al\mbox{.}(2026)]%
        {fmagent}
\bibfield{author}{\bibinfo{person}{Haoran Ding}, \bibinfo{person}{Zhaoguo
  Wang}, {and} \bibinfo{person}{Haibo Chen}.} \bibinfo{year}{2026}\natexlab{}.
\newblock \showarticletitle{{FM-Agent}: Scaling Formal Methods to Large Systems
  via {LLM}-Based Hoare-Style Reasoning}.
\newblock \bibinfo{journal}{\emph{CoRR}}  \bibinfo{volume}{abs/2604.11556}
  (\bibinfo{year}{2026}).
\newblock
\urldef\tempurl%
\url{https://doi.org/10.48550/ARXIV.2604.11556}
\showDOI{\tempurl}


\bibitem[Ding et~al\mbox{.}(2023)]%
        {flux}
\bibfield{author}{\bibinfo{person}{Haoran Ding}, \bibinfo{person}{Zhaoguo
  Wang}, \bibinfo{person}{Zhuohao Shen}, \bibinfo{person}{Rong Chen}, {and}
  \bibinfo{person}{Haibo Chen}.} \bibinfo{year}{2023}\natexlab{}.
\newblock \showarticletitle{Automated Verification of Idempotence for Stateful
  Serverless Applications}. In \bibinfo{booktitle}{\emph{OSDI}}.
\newblock


\bibitem[Flanagan and Freund(2009)]%
        {flanagan-fasttrack}
\bibfield{author}{\bibinfo{person}{Cormac Flanagan} {and}
  \bibinfo{person}{Stephen~N. Freund}.} \bibinfo{year}{2009}\natexlab{}.
\newblock \showarticletitle{{FastTrack}: Efficient and Precise Dynamic Race
  Detection}. In \bibinfo{booktitle}{\emph{PLDI}}.
\newblock


\bibitem[{Google}(2024)]%
        {xla}
\bibfield{author}{\bibinfo{person}{{Google}}.} \bibinfo{year}{2024}\natexlab{}.
\newblock \bibinfo{title}{XLA: Optimizing Compiler for Machine Learning}.
\newblock \bibinfo{howpublished}{Developer Documentation}.
\newblock


\bibitem[Herklotz et~al\mbox{.}(2021)]%
        {verilog-hls}
\bibfield{author}{\bibinfo{person}{Yann Herklotz}, \bibinfo{person}{James~D.
  Pollard}, \bibinfo{person}{Nadesh Ramanathan}, {and} \bibinfo{person}{John
  Wickerson}.} \bibinfo{year}{2021}\natexlab{}.
\newblock \showarticletitle{Formal Verification of High-Level Synthesis}. In
  \bibinfo{booktitle}{\emph{OOPSLA}}.
\newblock


\bibitem[Holzmann(1997)]%
        {spin}
\bibfield{author}{\bibinfo{person}{Gerard~J. Holzmann}.}
  \bibinfo{year}{1997}\natexlab{}.
\newblock \showarticletitle{The Model Checker {SPIN}}.
\newblock \bibinfo{journal}{\emph{IEEE Transactions on Software Engineering}}
  \bibinfo{volume}{23}, \bibinfo{number}{5} (\bibinfo{year}{1997}),
  \bibinfo{pages}{279--295}.
\newblock


\bibitem[Huawei(2023)]%
        {ascendc}
\bibfield{author}{\bibinfo{person}{Huawei}.} \bibinfo{year}{2023}\natexlab{}.
\newblock \bibinfo{title}{Ascend {C} Programming Guide}.
\newblock \bibinfo{howpublished}{Huawei Developer Documentation}.
\newblock


\bibitem[{Huawei Technologies}(2025)]%
        {mssanitizer}
\bibfield{author}{\bibinfo{person}{{Huawei Technologies}}.}
  \bibinfo{year}{2025}\natexlab{}.
\newblock \bibinfo{title}{{msSanitizer}: Anomaly Detection for {AscendC}
  Operators}.
\newblock \bibinfo{howpublished}{CANN 8.0.RC3 Developer Documentation}.
\newblock
\newblock
\shownote{\url{https://www.hiascend.com/document/detail/zh/canncommercial/80RC3/devaids/opdev/optool/atlasopdev_16_0039.html}}.


\bibitem[Jia and Harman(2011)]%
        {jia-mutation}
\bibfield{author}{\bibinfo{person}{Yue Jia} {and} \bibinfo{person}{Mark
  Harman}.} \bibinfo{year}{2011}\natexlab{}.
\newblock \showarticletitle{An Analysis and Survey of the Development of
  Mutation Testing}.
\newblock \bibinfo{journal}{\emph{IEEE TSE}} \bibinfo{volume}{37},
  \bibinfo{number}{5} (\bibinfo{year}{2011}).
\newblock


\bibitem[Jia et~al\mbox{.}(2019)]%
        {graphcore-ipu}
\bibfield{author}{\bibinfo{person}{Zhe Jia}, \bibinfo{person}{Blake Tillman},
  \bibinfo{person}{Marco Maggioni}, {and} \bibinfo{person}{Daniele~Paolo
  Scarpazza}.} \bibinfo{year}{2019}\natexlab{}.
\newblock \showarticletitle{Dissecting the {Graphcore IPU} Architecture via
  Microbenchmarking}.
\newblock \bibinfo{journal}{\emph{CoRR}}  \bibinfo{volume}{abs/1912.03413}.
\newblock


\bibitem[Jouppi et~al\mbox{.}(2017)]%
        {jouppi-tpu}
\bibfield{author}{\bibinfo{person}{Norman~P. Jouppi}, \bibinfo{person}{Cliff
  Young}, \bibinfo{person}{Nishant Patil}, \bibinfo{person}{David Patterson},
  {et~al\mbox{.}}} \bibinfo{year}{2017}\natexlab{}.
\newblock \showarticletitle{In-Datacenter Performance Analysis of a Tensor
  Processing Unit}.
\newblock  (\bibinfo{year}{2017}).
\newblock


\bibitem[Lamport(1978)]%
        {lamport}
\bibfield{author}{\bibinfo{person}{Leslie Lamport}.}
  \bibinfo{year}{1978}\natexlab{}.
\newblock \showarticletitle{Time, Clocks, and the Ordering of Events in a
  Distributed System}.
\newblock \bibinfo{journal}{\emph{Commun. ACM}} \bibinfo{volume}{21},
  \bibinfo{number}{7} (\bibinfo{year}{1978}).
\newblock


\bibitem[Lamport(1994)]%
        {tla-lamport}
\bibfield{author}{\bibinfo{person}{Leslie Lamport}.}
  \bibinfo{year}{1994}\natexlab{}.
\newblock \showarticletitle{The Temporal Logic of Actions}.
\newblock \bibinfo{journal}{\emph{ACM TOPLAS}} \bibinfo{volume}{16},
  \bibinfo{number}{3}, \bibinfo{pages}{872--923}.
\newblock


\bibitem[Leroy(2009)]%
        {compcert}
\bibfield{author}{\bibinfo{person}{Xavier Leroy}.}
  \bibinfo{year}{2009}\natexlab{}.
\newblock \showarticletitle{Formal Verification of a Realistic Compiler}.
\newblock \bibinfo{journal}{\emph{Commun. ACM}} (\bibinfo{year}{2009}).
\newblock


\bibitem[Li and Gopalakrishnan(2010)]%
        {li-gpu-determinism}
\bibfield{author}{\bibinfo{person}{Guodong Li} {and} \bibinfo{person}{Ganesh
  Gopalakrishnan}.} \bibinfo{year}{2010}\natexlab{}.
\newblock \showarticletitle{Scalable {SMT}-Based Verification of {GPU} Kernel
  Functions}. In \bibinfo{booktitle}{\emph{FSE}}.
\newblock


\bibitem[Li et~al\mbox{.}(2012)]%
        {gklee}
\bibfield{author}{\bibinfo{person}{Guodong Li}, \bibinfo{person}{Peng Li},
  \bibinfo{person}{Ganesh Gopalakrishnan}, \bibinfo{person}{Shan Lu}, {and}
  \bibinfo{person}{Jeremy Cong}.} \bibinfo{year}{2012}\natexlab{}.
\newblock \showarticletitle{{GKLEE}: Concolic Verification and Test Generation
  for GPUs}. In \bibinfo{booktitle}{\emph{PLDI}}.
\newblock


\bibitem[Li et~al\mbox{.}(2014a)]%
        {racefuzzer}
\bibfield{author}{\bibinfo{person}{Peng Li}, \bibinfo{person}{Ganesh
  Gopalakrishnan}, {et~al\mbox{.}}} \bibinfo{year}{2014}\natexlab{a}.
\newblock \showarticletitle{Effective Random Testing for Detecting Concurrency
  Bugs in {GPU} Programs}. In \bibinfo{booktitle}{\emph{ASE}}.
\newblock


\bibitem[Li et~al\mbox{.}(2014b)]%
        {symbolic-gpu}
\bibfield{author}{\bibinfo{person}{Peng Li}, \bibinfo{person}{Chang Liu}, {and}
  \bibinfo{person}{Ganesh Gopalakrishnan}.} \bibinfo{year}{2014}\natexlab{b}.
\newblock \showarticletitle{Practical Symbolic Race Checking of {GPU}
  Programs}. In \bibinfo{booktitle}{\emph{SC}}.
\newblock


\bibitem[Li et~al\mbox{.}(2020)]%
        {simulee}
\bibfield{author}{\bibinfo{person}{Ying Li}, \bibinfo{person}{Xiang Gao},
  {et~al\mbox{.}}} \bibinfo{year}{2020}\natexlab{}.
\newblock \showarticletitle{{Simulee}: Detecting CUDA Synchronization Bugs via
  Dynamic Analysis}. In \bibinfo{booktitle}{\emph{ICSE}}.
\newblock


\bibitem[Lopes et~al\mbox{.}(2015)]%
        {alive}
\bibfield{author}{\bibinfo{person}{Nuno~P. Lopes}, \bibinfo{person}{David
  Menendez}, \bibinfo{person}{Santosh Nagarakatte}, {and} \bibinfo{person}{John
  Regehr}.} \bibinfo{year}{2015}\natexlab{}.
\newblock \showarticletitle{Provably Correct Peephole Optimizations with
  {Alive}}. In \bibinfo{booktitle}{\emph{PLDI}}.
\newblock


\bibitem[Lustig et~al\mbox{.}(2019)]%
        {ptxmem}
\bibfield{author}{\bibinfo{person}{Daniel Lustig}, \bibinfo{person}{Sameer
  Sahasrabuddhe}, {and} \bibinfo{person}{Olivier Giroux}.}
  \bibinfo{year}{2019}\natexlab{}.
\newblock \showarticletitle{A Formal Analysis of the {NVIDIA PTX} Memory
  Consistency Model}. In \bibinfo{booktitle}{\emph{ASPLOS}}.
\newblock


\bibitem[Min\'{e}(2006)]%
        {mine-octagon}
\bibfield{author}{\bibinfo{person}{Antoine Min\'{e}}.}
  \bibinfo{year}{2006}\natexlab{}.
\newblock \showarticletitle{The Octagon Abstract Domain}. In
  \bibinfo{booktitle}{\emph{Higher-Order and Symbolic Computation}},
  Vol.~\bibinfo{volume}{19}. \bibinfo{pages}{31--100}.
\newblock


\bibitem[Netzer and Miller(1992)]%
        {netzer-miller}
\bibfield{author}{\bibinfo{person}{Robert H.~B. Netzer} {and}
  \bibinfo{person}{Barton~P. Miller}.} \bibinfo{year}{1992}\natexlab{}.
\newblock \showarticletitle{What Are Race Conditions? Some Issues and
  Formalizations}.
\newblock \bibinfo{journal}{\emph{ACM Lett. Program. Lang. Syst.}}
  \bibinfo{volume}{1}, \bibinfo{number}{1} (\bibinfo{year}{1992}),
  \bibinfo{pages}{74--88}.
\newblock


\bibitem[{NVIDIA}(2024a)]%
        {cuda-racecheck}
\bibfield{author}{\bibinfo{person}{{NVIDIA}}.}
  \bibinfo{year}{2024}\natexlab{a}.
\newblock \bibinfo{title}{CUDA-MEMCHECK Racecheck Tool Documentation}.
\newblock \bibinfo{howpublished}{NVIDIA Developer Documentation}.
\newblock


\bibitem[{NVIDIA}(2024b)]%
        {tensorrt}
\bibfield{author}{\bibinfo{person}{{NVIDIA}}.}
  \bibinfo{year}{2024}\natexlab{b}.
\newblock \bibinfo{title}{TensorRT Developer Guide}.
\newblock \bibinfo{howpublished}{Developer Documentation}.
\newblock


\bibitem[Ragan-Kelley et~al\mbox{.}(2013)]%
        {halide}
\bibfield{author}{\bibinfo{person}{Jonathan Ragan-Kelley},
  \bibinfo{person}{Andrew Adams}, \bibinfo{person}{Sylvain Paris},
  \bibinfo{person}{Marc Levoy}, \bibinfo{person}{Saman Amarasinghe}, {and}
  \bibinfo{person}{Fr{\'e}do Durand}.} \bibinfo{year}{2013}\natexlab{}.
\newblock \showarticletitle{Halide: Decoupling Algorithms from Schedules for
  High-Performance Image Processing}. In \bibinfo{booktitle}{\emph{PLDI}}.
\newblock


\bibitem[Reps et~al\mbox{.}(1995)]%
        {reps-horwitz-sagiv}
\bibfield{author}{\bibinfo{person}{Thomas Reps}, \bibinfo{person}{Susan
  Horwitz}, {and} \bibinfo{person}{Mooly Sagiv}.}
  \bibinfo{year}{1995}\natexlab{}.
\newblock \showarticletitle{Precise Interprocedural Dataflow Analysis via Graph
  Reachability}.
\newblock  (\bibinfo{year}{1995}).
\newblock


\bibitem[Savage et~al\mbox{.}(1997)]%
        {eraser}
\bibfield{author}{\bibinfo{person}{Stefan Savage}, \bibinfo{person}{Michael
  Burrows}, \bibinfo{person}{Greg Nelson}, \bibinfo{person}{Patrick
  Sobalvarro}, {and} \bibinfo{person}{Thomas Anderson}.}
  \bibinfo{year}{1997}\natexlab{}.
\newblock \showarticletitle{Eraser: A Dynamic Data Race Detector for
  Multithreaded Programs}.
\newblock \bibinfo{journal}{\emph{ACM Trans. Comput. Syst.}}
  \bibinfo{volume}{15}, \bibinfo{number}{4}, \bibinfo{pages}{391--411}.
\newblock


\bibitem[Serebryany and Iskhodzhanov(2009)]%
        {tsan}
\bibfield{author}{\bibinfo{person}{Konstantin Serebryany} {and}
  \bibinfo{person}{Timur Iskhodzhanov}.} \bibinfo{year}{2009}\natexlab{}.
\newblock \showarticletitle{{ThreadSanitizer}: Data Race Detection in
  Practice}. In \bibinfo{booktitle}{\emph{WBIA (Workshop on Binary
  Instrumentation and Applications)}}.
\newblock


\bibitem[Tillet et~al\mbox{.}(2019)]%
        {triton}
\bibfield{author}{\bibinfo{person}{Philippe Tillet}, \bibinfo{person}{H.~T.
  Kung}, {and} \bibinfo{person}{David Cox}.} \bibinfo{year}{2019}\natexlab{}.
\newblock \showarticletitle{Triton: An Intermediate Language and Compiler for
  Tiled Neural Network Computations}. In \bibinfo{booktitle}{\emph{MLSys}}.
\newblock


\bibitem[Wickerson et~al\mbox{.}(2017)]%
        {wickerson-memalloy}
\bibfield{author}{\bibinfo{person}{John Wickerson}, \bibinfo{person}{Mark
  Batty}, \bibinfo{person}{Tyler Sorensen}, {and} \bibinfo{person}{George~A.
  Constantinides}.} \bibinfo{year}{2017}\natexlab{}.
\newblock \showarticletitle{Automatically Comparing Memory Consistency Models}.
  In \bibinfo{booktitle}{\emph{POPL}}.
\newblock


\bibitem[Zhao et~al\mbox{.}(2021)]%
        {akg}
\bibfield{author}{\bibinfo{person}{Jie Zhao}, \bibinfo{person}{Bojie Li},
  \bibinfo{person}{Wang Nie}, \bibinfo{person}{Zhen Geng},
  \bibinfo{person}{Renwei Zhang}, \bibinfo{person}{Xiong Gao},
  \bibinfo{person}{Bin Cheng}, \bibinfo{person}{Chen Wu}, \bibinfo{person}{Yun
  Cheng}, \bibinfo{person}{Zheng Li}, \bibinfo{person}{Peng Di},
  \bibinfo{person}{Kun Zhang}, {and} \bibinfo{person}{Xuefeng Jin}.}
  \bibinfo{year}{2021}\natexlab{}.
\newblock \showarticletitle{{AKG}: Automatic Kernel Generation for Neural
  Processing Units using Polyhedral Transformations}. In
  \bibinfo{booktitle}{\emph{PLDI}}.
\newblock


\bibitem[Zheng et~al\mbox{.}(2020)]%
        {ansor}
\bibfield{author}{\bibinfo{person}{Lianmin Zheng}, \bibinfo{person}{Chengfan
  Jia}, \bibinfo{person}{Minmin Sun}, \bibinfo{person}{Zhao Wu},
  \bibinfo{person}{Cody~Hao Yu}, \bibinfo{person}{Ameer Haj-Ali},
  \bibinfo{person}{Yida Wang}, \bibinfo{person}{Jun Yang},
  \bibinfo{person}{Danyang Zhuo}, \bibinfo{person}{Koushik Sen},
  \bibinfo{person}{Joseph~E. Gonzalez}, {and} \bibinfo{person}{Ion Stoica}.}
  \bibinfo{year}{2020}\natexlab{}.
\newblock \showarticletitle{Ansor: Generating High-Performance Tensor Programs
  for Deep Learning}. In \bibinfo{booktitle}{\emph{OSDI}}.
\newblock


\bibitem[Zhu et~al\mbox{.}(2022)]%
        {roller}
\bibfield{author}{\bibinfo{person}{Hongyu Zhu}, \bibinfo{person}{Ruofan Wu},
  \bibinfo{person}{Yijia Diao}, \bibinfo{person}{Shanbin Ke},
  \bibinfo{person}{Haoyu Li}, \bibinfo{person}{Chen Zhang},
  \bibinfo{person}{Jilong Xue}, \bibinfo{person}{Lingxiao Ma},
  \bibinfo{person}{Yuqing Xia}, \bibinfo{person}{Wei Cui}, \bibinfo{person}{Fan
  Yang}, \bibinfo{person}{Mao Yang}, \bibinfo{person}{Lidong Zhou},
  \bibinfo{person}{Asaf Cidon}, {and} \bibinfo{person}{Gennady Pekhimenko}.}
  \bibinfo{year}{2022}\natexlab{}.
\newblock \showarticletitle{{ROLLER}: Fast and Efficient Tensor Compilation for
  Deep Learning}. In \bibinfo{booktitle}{\emph{OSDI}}.
\newblock


\end{thebibliography}
\end{document}